\documentclass[a4paper,twocolumn,11pt]{article}
\pdfoutput=1
\usepackage[tmargin=0.75in, lmargin=0.6in, rmargin=0.6in, bmargin=1in]{geometry}
\usepackage{tikz}
\usepackage{amsmath}
\usepackage{amsfonts}
\usepackage{csquotes}
\usepackage{cite}
\usepackage{amssymb}
\usepackage{mathtools}
\usepackage{amsthm}
\usepackage{enumerate}
\usepackage{dsfont}
\usepackage{graphicx}
\usepackage{authblk}
\usepackage{hyperref}
\usepackage{epstopdf}
\usepackage{theoremref}
\usepackage{subfig}
\usepackage{bbm}
\usepackage{bm}
\usepackage{qcircuit}
\usepackage{diagbox}
\usetikzlibrary{arrows,decorations.pathmorphing,backgrounds,positioning,fit}
\usetikzlibrary{positioning,chains,fit,shapes,calc,math,arrows.meta}
\tikzstyle{vertex}=[circle, draw, inner sep=0pt, minimum size=6pt]

\usepackage{algorithm}
\usepackage{algpseudocode}
\usepackage{wrapfig}

\usepackage{blkarray, bigstrut}
\usepackage{booktabs}
\usepackage[stable]{footmisc}

\newtheorem{theorem}{Theorem}
\newtheorem{definition}[theorem]{Definition}

\newtheorem{remark}[theorem]{Remark}
\newtheorem{example}{Example}

\newcommand{\llbr}{[\![}
\newcommand{\rrbr}{]\!]}

\newcommand{\C}{\mathbb C}

\newcommand{\Z}{\mathbb Z}
\newcommand{\F}{\mathbb F}

\newcommand\numberthis{\addtocounter{equation}{1}\tag{\theequation}}
\newcommand{\ket}[1]{|{#1}\rangle}
\newcommand{\bra}[1]{\langle{#1}|}

\usepackage{authblk}

\makeatletter
\newbox\abstract@box
\renewenvironment{abstract}
{\global\setbox\abstract@box=\vbox\bgroup
	\hsize=\textwidth\linewidth=\textwidth
	\small
	\begin{center}%
		{\bfseries \abstractname\vspace{-.5em}\vspace{\z@}}%
	\end{center}%
	\quotation}
{\endquotation\egroup}
\expandafter\def\expandafter\@maketitle\expandafter{\@maketitle
	\ifvoid\abstract@box\else\unvbox\abstract@box\if@twocolumn\vskip1.5em\fi\fi}
\makeatother

\newif\ifnotes
\notesfalse

\begin{document}
	\title{Climbing the Diagonal Clifford Hierarchy}
	\author[1]{Jingzhen Hu\thanks{The first two authors contributed equally to this work.}$^,$}
	\author[1]{Qingzhong Liang$^{\ast,}$}
	\author[1,2,3]{Robert Calderbank}
	\affil[1]{Department of Mathematics, Duke University}
	\affil[2]{Department of Electrical and Computer Engineering, Duke University}
	\affil[3]{Department of Computer Science, Duke University, Durham, NC 27708, USA}
	\date{E-mail: \{jingzhen.hu, qingzhong.liang, robert.calderbank\}@duke.edu}
	\begin{abstract}
		Magic state distillation and the Shor factoring algorithm make essential use of logical diagonal gates. We introduce a method of synthesizing CSS codes that realize a target logical diagonal gate at some level $l$ in the Clifford hierarchy. The method combines three basic operations: concatenation, removal of $Z$-stabilizers, and addition of $X$-stabilizers. It explicitly tracks the logical gate induced by a diagonal physical gate that preserves a CSS code. The first step is concatenation, where the input is a CSS code and a physical diagonal gate at level $l$ inducing a logical diagonal gate at the same level. The output is a new code for which a physical diagonal gate at level $l+1$ induces the original logical gate. The next step is judicious removal of $Z$-stabilizers to increase the level of the induced logical operator. We identify three ways of climbing the logical Clifford hierarchy from level $l$ to level $l+1$, each built on a recursive relation on the Pauli coefficients of the induced logical operators. Removal of $Z$-stabilizers may reduce distance, and the purpose of the third basic operation, addition of $X$-stabilizers, is to compensate for such losses. For the coherent noise model, we describe how to switch between computation and storage of intermediate results in a decoherence-free subspace by simply applying Pauli $X$ matrices.
		The approach to logical gate synthesis taken in prior work focuses on the code states, and results in sufficient conditions for a CSS code to be fixed by a transversal $Z$-rotation. In contrast, we derive necessary and sufficient conditions by analyzing the action of a transversal diagonal gate on the stabilizer group that determines the code. The power of our approach to logical gate synthesis is demonstrated by two proofs of concept: the $\llbr 2^{l+1}-2,2,2\rrbr$ triorthogonal code family, and the $\llbr 2^m,\binom{m}{r},2^{\min\{r,m-r\}}\rrbr$ quantum Reed-Muller code family.

	\end{abstract}
	\maketitle
	\section{Introduction} 
	\label{sec:intro}
	The challenge of quantum computing is to combine error resilience with universal computation. There are many finite sets of gates that are universal, and a standard choice is to augment the set of Clifford gates by a non-Clifford unitary \cite{boykin1999universal} such as the $T$ gate $\left(T=Z^{1/4}\right)$. Gottesman and Chuang \cite{gottesman1999demonstrating} defined the \emph{Clifford hierarchy} when introducing the teleportation model of quantum computing. The first level is the \emph{Pauli group}. The second level is the \emph{Clifford group}, which consists of unitary operators that normalize the Pauli group. The $l^{\mathrm{th}}$ level consists of unitary operators that map Pauli operators to the $(l-1)^{\mathrm{th}}$ level under conjugation. The structure of the Clifford hierarchy has been studied extensively \cite{zeng2008semi,beigi2008c3,bengtsson2014order,cui2017diagonal,rengaswamy2019unifying,pllaha2020weyl}. For $l \ge 3$, the operators at level $l$ are not closed under matrix multiplication. However, the diagonal gates at each level $l$ of the hierarchy do form a group \cite{zeng2008semi,cui2017diagonal}, and the gates $Z^{1/2^{l-1}}$, $\mathrm{C}^{(i)}Z^{1/2^{j}}$ with $i+j=l-1$ generate this group \cite{zeng2008semi}. The generators at the next level $l+1$ can be obtained by taking a square root $\left(Z^{1/2^{l-1}}\to Z^{1/2^{l}}\right)$ or adding one more layer of control $\left(\mathrm{C}^{(i)}Z^{1/2^{j}} \to \mathrm{C}^{(i+1)}Z^{1/2^{j}}\right)$ as shown in Figure \ref{fig:elem_DCH}.
	
	Quantum error-correcting codes (QECCs) encode logical qubits into physical qubits, and protect information as it is transformed by logical gates. Given a logical diagonal operator among the generators of the diagonal Clifford hierarchy, we describe a general method for synthesizing a CSS code \cite{Calderbank-physreva96,Steane-physreva96} preserved by a diagonal physical gate which induces the target logical operator. Logical diagonal gates play a central role in quantum algorithms. In the Shor factoring algorithm \cite{shor1994algorithms,shor1999polynomial}, our method applies to the C$^{(i)}Z^{1/2^j}$ diagonal gates which play an essential role in period finding. In magic state distillation (MSD) \cite{bravyi2005universal,reichardt2005quantum,bravyi2012magic,anwar2012qutrit,campbell2012magic,landahl2013complex,campbell2017unified,haah2018codes,krishna2019towards,vuillot2019quantum}, the effectiveness of the protocol depends on engineering the interaction of a diagonal physical gate with the code states of a stabilizer code \cite{gottesman1997stabilizer,calderbank1998quantum}. Our method transforms a CSS code supporting a lower level logical operator to a CSS code supporting a higher level logical operator. The coefficients in the Pauli expansion of a diagonal gate satisfy a recursion that makes it possible to work backwards from a target logical gate.
	
	Throughout the paper, we make use of an explicit representation of the logical channel induced by a diagonal physical gate. We prepare an initial state, apply a physical gate, then measure a code syndrome $\bm{\mu}$, and finally apply a correction based on $\bm{\mu}$. For each syndrome, we expand the induced logical operator in the Pauli basis to obtain the \emph{generator coefficients} \cite{generator_coeff_framework} that capture state evolution. Intuitively, the diagonal physical gate preserves the code space if and only if the induced logical operator corresponding to the trivial syndrome is unitary. To support the objective of fault tolerance, we emphasize \emph{transversal} gates \cite{gottesman1997stabilizer}, which are tensor products of unitaries on individual code blocks. The approach taken in prior work is to focus on the code states, and to derive sufficient conditions for a stabilizer code to be fixed by a transversal $Z$-rotation\cite{bravyi2005universal,reichardt2005quantum,bravyi2012magic,campbell2012magic,landahl2013complex,campbell2017unified,haah2018codes,haah2018towers,vuillot2019quantum}.  In contrast we derive necessary and sufficient conditions by analyzing the action of a transversal diagonal gate on the stabilizer group that determines the code. An advantage of our approach is that we keep track of the induced logical operator. 
	
	The action of a diagonal physical operator $U_Z$ on code states depends very strongly on the signs of $Z$-stabilizers \cite{coherent_noise,debroy2021optimizing,generator_coeff_framework} and our generator coefficient framework captures how these signs change the logical operators induced by $U_Z$. For the coherent noise model, a judicious choice of signs creates a decoherence-free subspace, that enables data storage. We can switch between computation and storage by applying a Pauli matrix as described in Remark \ref{rem:sw_comp_storage}.
	
	Haah \cite{haah2018towers} used divisibility properties of classical codes to construct CSS codes with parameters $\llbr O(d^{l-1}),\Omega(d),d \rrbr$ that realize a transversal logical $Z^{1/2^{l-1}}$. Modulo Clifford gates, his construction includes the $\llbr 2^l,1,3 \rrbr$ punctured quantum Reed-Muller (QRM) codes \cite{landahl2013complex} that support a single logical $Z^{1/2^{l-2}}$ gate, and the family of $\llbr 6k+8,2k,2 \rrbr$ \emph{triorthogonal} code \cite{bravyi2012magic} that support a logical transversal $T$ gate. In contrast we introduce three basic operations - concatenation, removal of $Z$-stabilizers, and addition of $X$-stabilizers - that can be combined to synthesize an arbitrary logical diagonal gate. We present the 
	$\llbr 2^m, \binom{m}{r},2^{\min\{r,m-r\}}\rrbr$ QRM code family \cite{rengaswamy2020optimality,generator_coeff_framework} as a proof of concept. 
	
	\begin{figure}[!ht]
		\resizebox{\linewidth}{.5\linewidth}{%
			\begin{tikzpicture}
				\hspace{5pt}
				\node (phy) at (2.5,-0.3) {physical level};
				\node (log) at (5.3,2.5) {\rotatebox{270}{logical level}};
				\draw[thick,dashed] (5,5) -- (5,0) -- (0,0)  -- (5,5);
				\draw[->,blue] (0.5,0.2) -- (1.45,0.2);
				\draw[->,red] (1.5,0.25) -- (1.5,1.2);
				
				\draw[->,blue] (6,4.5) -- (6.95,4.5);
				\draw[->,red] (6.5,3) -- (6.5,4);
				\filldraw[very thick] (6.5,2.5) circle (0.075);
				\node (con) at (8.5,4.5) {\color{blue}Concatenation};
				\node (rem) at (9.25,3.5) {\color{red}Removing $Z$-stabilizers};
				\node (add) at (9.1,2.5) {Adding $X$-stabilizers};
				\node (exam) at (7,1.5) {\textbf{Example}:};
				\node (steane) at (7,-0.5) {\llbr 4,2,2\rrbr};
				\node (con_steane) at (9.5,-0.5) {\llbr 64,2,2\rrbr};
				\node (tri_orth) at (9.5,1) {\llbr 64,15,4\rrbr};
				\draw[->,blue] (steane) -- (con_steane);
				\draw[->,red] (9.25,-0.25) -- (9.25,0.75); 
				\filldraw[very thick] (9.8,0.25) circle (0.07);
				\node () at (9.5,0.25) {$+$};
				\draw[->,dashed] (7.5,-0.25) -- (8.75,1);
			\end{tikzpicture}
		}
		\caption{Three basic operations that can be combined to synthesize a CSS code with higher distance, preserved by a diagonal physical gate which induces a prescribed logical diagonal gate.}
		\label{fig:three_elem_ops}
	\end{figure}
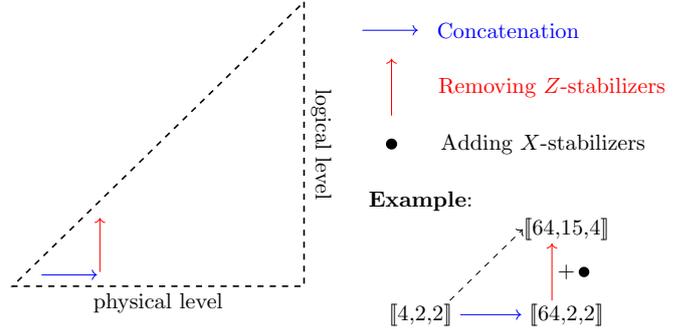
	
	Figure \ref{fig:three_elem_ops} shows how the three basic operations combine to provide CSS codes where both distance and the level of the induced logical operator are increasing. We now examine the three basic operations in more detail. 
	\begin{enumerate}
		\item \textbf{Concatenation}. Figure \ref{fig:three_elem_ops} shows that the level of the induced logical operator is bounded by that of the physical operator. Concatenation is depicted in Figure \ref{fig:concatenation} and described in Section \ref{sec:concatenation}. We double the number of physical qubits to increase the level of the physical diagonal gate and to make room for increasing the level of the induced logical operator. Theorem \ref{thm:concatenation} characterize the family of physical diagonal gates acting on the new code to induce the same logical gate. For example, the $\llbr 7,1,3\rrbr$ Steane code \cite{steane1996multiple} is preserved by a transversal Phase gate, $P^{\otimes 7}=\left({Z}^{1/2}\right)^{\otimes 7}$, which induces a logical $P^\dagger$ gate. By concatenating once, we obtain the $\llbr 14,1,3\rrbr$ CSS code that supports the logical $P^\dagger$ gate through a family of physical gates including the $I_2^{\otimes 7}\otimes P^{\otimes 7}$ physical gate at level $2$ and the transversal $T$ gate ($T^{\otimes 14}$) at level $3$. The higher level gate creates the opportunity to use the second basic operation to increase the level of the induced logical operator.
		
		\item \textbf{Removal of $Z$-stabilizers}. This is depicted in Figure \ref{fig:rem_Z_stab} and described in Section \ref{sec:rem_Z_stab}. We increase the code rate by removing a non-trivial $Z$-stabilizer to introduce a new logical qubit. Each generator coefficient in the expansion of the original logical operator splits into two new generator coefficients. We provide necessary and sufficient conditions for the new code to be preserved by the original physical diagonal gate. In this case we say that the removal/split is admissible. We describe three types of admissible split that increase the level of the induced logical operator, each built on a recursive relation on the generator coefficients. The two splits described in Figure \ref{fig:Z-rot_split_schemes} apply trigonometric identities. When the physical gate is a transversal $Z^{1/2^{l}}$, Theorem \ref{thm:specif_rem_Z} specifies the $Z$-stabilizer that is to be removed. For example, removing the all-one $Z$-stabilizer from the $\llbr 14,1,3 \rrbr$ code gives the $\llbr 14,2,2 \rrbr$ triorthogonal code, and the induced logical operator becomes a transversal $T^\dagger$. Distance may decrease after removing a $Z$-stabilizer, and the purpose of the third basic operation is to compensate this loss.
		
		\item \textbf{Addition of X-stabilizers}. This is depicted in Figure \ref{fig:add_X_stab} and described in Section \ref{sec:add_X_stab}. We derive  necessary and sufficient conditions for the new code after addition to be preserved by the original physical diagonal gate, and we say that the addition is admissible in this case. Our conditions require that half the generator coefficients associated with the trivial syndrome must vanish. For an admissible addition, we show that the level of the induced logical operator is unchanged. We may need to concatenate several times and to remove several independent $Z$-stabilizers in order to create sufficiently many zeros to enable an admissible addition. For example, consider the $\llbr 4,2,2\rrbr$ CSS code defined by the stabilizer group $\mathcal{S}=\langle X^{\otimes 4}, Z^{\otimes 4}\rangle$. Up to some logcial Pauli $Z$, the code realizes a logical C$Z$ by a transversal Phase gate. We first concatenate $4$ times to obtain the $\llbr 64,2,2\rrbr$ CSS code with the same logical operator, but induced by a physical transversal $T$ gate. Then, we remove $19$ independent $Z$-stabilizers to produce the $\llbr 64,21,2\rrbr$ code that realizes $15$ logical CC$Z$ gates (up to logical Pauli $Z$) induced by a physical transversal $T$ gate. Finally, we add $6$ independent $X$-stabilizers to increase the distance and arrive the $\llbr 64,15,4\rrbr$ QRM code supporting the same physical and logical gates. 
	\end{enumerate}
	
	The next Section introduces notation and provides necessary background. Section \ref{sec:concatenation}, \ref{sec:rem_Z_stab}, and \ref{sec:add_X_stab} introduce concatenation, removal of $Z$-stabilizers, and addition of $X$-stabilizers respectively. 

	\section{Preliminaries and Notation}
	\label{sec:prelims}
	\subsection{The Pauli Group}
	\label{subsec:perlim_Pauli}
	Let $\imath\coloneqq \sqrt{-1}$ be the imaginary unit. 
	Any $2\times 2$ Hermitian matrix can be uniquely expressed as a real linear combination of the four single qubit Pauli matrices/operators
	\begin{align}
		I_2 \coloneqq \begin{bmatrix}
			1 & 0\\
			0 & 1
		\end{bmatrix},~  
		X \coloneqq \begin{bmatrix} 
			0 & 1\\
			1 & 0
		\end{bmatrix},~  
		Z \coloneqq  \begin{bmatrix} 
			1 & 0\\
			0 & -1
		\end{bmatrix}, 
	\end{align}
	and $Y\coloneqq \imath XZ$.
	The operators satisfy 
	$
	X^2= Y^2= Z^2=I_2,~  X Y=- Y X,~  X Z=- Z X,~ \text{ and }  Y Z=- Z Y.
	$
	
	Let $\F_2 = \{0,1\}$ denote the binary field. Let $n\ge 1$ and $N=2^n$. Let $A \otimes B$ denote the Kronecker product (tensor product) of two matrices $A$ and $B$. Given binary vectors $\bm{a}=[a_1,a_2,\dots,a_n]$ and $\bm{b}=[b_1,b_2,\dots,b_n]$ with $a_i,b_j =0$ or $1$, we define the operators
	\begin{align}
		D(\bm{a},\bm{b})&\coloneqq X^{a_1} Z^{b_1}\otimes \cdots \otimes  X^{a_n} Z^{b_n},\\
		E(\bm{a},\bm{b}) &\coloneqq\imath^{\bm{a}\bm{b}^T \bmod 4}D(\bm{a},\bm{b}).
	\end{align}
	
	We often abuse notation and write $\bm{a}, \bm{b} \in \F_2^n$, though entries of vectors are sometimes interpreted in $\mathbb{Z}_4 = \{ 0,1,2,3 \}$. Note that $D(\bm{a},\bm{b})$ can have order $1,2$ or $4$, but $E(\bm{a},\bm{b})^2=\imath^{2\bm{a}\bm{b}^T}D(\bm{a},\bm{b})^2=\imath^{2ab^T}( \imath^{2\bm{a}\bm{b}^T} I_N)=I_N$. The $n$-qubit \textit{Pauli group} is defined as
	\begin{equation}
		\mathcal{HW}_N \coloneqq\{\imath^\kappa D(\bm{a},\bm{b}): \bm{a},\bm{b}\in \F_2^n, \kappa\in \Z_{4} \},
	\end{equation}
	where $\Z_{2^l} = \{0,1,\dots,2^l-1\}$.
	The $n$-qubit Pauli matrices form an orthonormal basis for the vector space of $N\times N$ complex matrices ($\C^{N\times N}$) under the normalized Hilbert-Schmidt inner product $\langle A,B\rangle \coloneqq \mathrm{Tr}(A^\dagger B)/N$ \cite{gottesman1997stabilizer}. 
	
	We use the \textit{Dirac notation}, $|\cdot \rangle$ to represent the basis states of a single qubit in $\C^2$. For any $\bm{v}=[v_1,v_2,\cdots, v_n]\in \F_2^n$, we define $|\bm{v}\rangle=|v_1\rangle\otimes|v_2\rangle\otimes\cdots\otimes|v_n\rangle$, the standard basis vector in $\C^N$ with $1$ in the position indexed by $\bm{v}$ and $0$ elsewhere. 
	We write the Hermitian transpose of $|\bm{v}\rangle$ as $\langle \bm{v}|=|\bm{v}\rangle^\dagger$. 
	We may write an arbitrary $n$-qubit quantum state as $|\psi\rangle=\sum_{\bm{v}\in \F_2^n} \alpha_{\bm{v}} |\bm{v}\rangle \in \C^N$, where $\alpha_{\bm{v}}\in \C$ and $\sum_{\bm{v}\in\F_2^n}|\alpha_{\bm{v}}|^2=1$. The Pauli matrices act on a single qubit as
	$X\ket{0}=\ket{1},  X\ket{1}=\ket{0},  Z\ket{0}=\ket{0}, \text{ and }  Z\ket{1}=-\ket{1}.$
	
	The symplectic inner product is $\langle [\bm{a},\bm{b}],[\bm{c},\bm{d}]\rangle_S=\bm{a}\bm{d}^T+\bm{b}\bm{c}^T \bmod 2$. Since $ X Z=- Z X$, we have 
	\begin{equation}
		E(\bm{a},\bm{b})E(\bm{c},\bm{d})=(-1)^{\langle [\bm{a},\bm{b}],[\bm{c},\bm{d}]\rangle_S}E(\bm{c},\bm{d})E(\bm{a},\bm{b}).
	\end{equation}
	
	\subsection{The Clifford Hierarchy}
	\label{subsec:perlim_Clifford}
	The \textit{Clifford hierarchy} of unitary operators was introduced in \cite{gottesman1999demonstrating}. The first level of the hierarchy is defined to be the Pauli group $\mathcal{C}^{(1)}=\mathcal{HW}_N$. For $l\ge 2$, the levels $l$ are defined recursively as 
	\begin{equation}\label{eqn:def_Cliff_hierarchy}
		\mathcal{C}^{(l)}:=\{U\in \mathbb{U}_N: U \mathcal{HW}_N U^\dagger\subset \mathcal{C}^{(l-1)}\},
	\end{equation}
	where $\mathbb{U}_N$ is the group of $N\times N$ unitary matrices. The second level is the Clifford Group, $\mathcal{C}^{(2)}$, which can be generated (up to overall phases) using the ``elementary" unitaries \textit{Hadamard}, \textit{Phase}, and either of \textit{Controlled-NOT} (C$X$) or \textit{Controlled-$Z$} (C$Z$) defined respectively as
	\begin{equation}
		H\coloneqq \begin{bmatrix}
			1 & 1\\
			1 & -1
		\end{bmatrix} , 
		P\coloneqq\begin{bmatrix}
			1 & 0\\
			0 & \imath
		\end{bmatrix} , 
	\end{equation}
	\begin{align}
		\text{C}Z_{ab}  &\coloneqq \ket{0}\bra{0}_a \otimes (I_2)_b + \ket{1}\bra{1}_a \otimes Z_b,\\~
		\text{C}X_{a \rightarrow b}  &\coloneqq \ket{0}\bra{0}_a \otimes (I_2)_b + \ket{1}\bra{1}_a \otimes X_b.
	\end{align}
	
	Note that Clifford unitaries in combination with \emph{any} unitary from a higher level can be used to approximate any unitary operator arbitrarily well~\cite{boykin1999universal}. 
	Hence, they form a universal set for quantum computation. A widely used choice for the non-Clifford unitary is the $T$ gate in the third level defined by 
	\begin{equation}
		T:=\begin{bmatrix}
			1 & 0\\
			0 & e^{\frac{\imath\pi}{4}}
		\end{bmatrix}=
		\sqrt{P}=Z^{\frac{1}{4}}\equiv 
		\begin{bmatrix}
			e^{-\frac{\imath\pi}{8}} & 0\\
			0 & e^{\frac{\imath\pi}{8}}
		\end{bmatrix}
		=e^{-\frac{\imath\pi}{8} Z}.
	\end{equation}
	Let $\mathcal{D}_N$ be the $N\times N$ diagonal matrices, and $\mathcal{C}^{(l)}_d \coloneqq \mathcal{C}^{(l)} \cap \mathcal{D}_N$. While $\mathcal{C}^{(l)}$ for $l\ge 3$ do not form a group any more, the diagonal gates in each level of the hierarchy, $\mathcal{C}^{(l)}_d$, form a group.  Note that $\mathcal{C}^{(l)}_d$ can be generated using the ``elementary" unitaries  C$^{(0)}Z^{\frac{1}{2^{l}}}$, C$^{(1)}Z^{\frac{1}{2^{l-1}}}, \dots, $C$^{(l-2)}Z^{\frac{1}{2}}$, C$^{(l-1)} Z$ \cite{zeng2008semi}, where C$^{(i)}Z^{\frac{1}{2^j}} \coloneqq \sum_{\bm{u}\in \F_2^{i+1}}\ket{\bm{u}}\bra{\bm{u}} +e^{\imath \frac{\pi}{2^j}}\ket{\bm{1}}\bra{\bm{1}}$ and here $\bm{1}\in\F_2^{i+1}$ is the vector consists of all ones.

	\subsection{Stabilizer Codes}
	\label{subsec:perlim_stab}
	We define a stabilizer group $\mathcal{S}$ to be a commutative subgroup of the Pauli group $\mathcal{HW}_N$, where every group element is Hermitian and no group element is $-I_N$. We say $\mathcal{S}$ has dimension $r$ if it can be generated by $r$ independent elements as $\mathcal{S}=\langle \nu_i E(\bm{c_i},\bm{d_i}): i=1,2,\dots, r \rangle$, where $\nu_i\in\{\pm1\}$ and $\bm{c_i},\bm{d_i}\in \F_2^n$. Since $\mathcal{S}$ is commutative, we must have $\langle [\bm{c_i},\bm{d_i}],[\bm{c_j},\bm{d_j}]\rangle_S=\bm{c_i}\bm{d_j}^T+\bm{d_i}\bm{c_j}^T=0\bmod 2$.
	
	Given a stabilizer group $\mathcal{S}$, the corresponding \textit{stabilizer code} is the fixed subspace $\mathcal{V}(\mathcal{S)}:=\{|\psi\rangle \in \C^N: g|\psi\rangle=|\psi\rangle \text{ for all } g\in \mathcal{S} \}$. 
	We refer to the subspace $\mathcal{V}(\mathcal{S})$ as an $\left[\left[n,k,d\right]\right]$ stabilizer code because it encodes $k:=n-r$ \text{logical} qubits into $n$ \textit{physical} qubits. The minimum distance $d$ is defined to be the minimum weight of any operator in $\mathcal{N}_{\mathcal{HW}_N}\left(\mathcal{S}\right)\setminus \mathcal{S}$. Here, the weight of a Pauli operator is the number of qubits on which it acts non-trivially (i.e., as $ X,~Y$ or $ Z$), and $\mathcal{N}_{\mathcal{HW}_N}\left(\mathcal{S}\right)$ denotes the normalizer of $\mathcal{S}$ in $\mathcal{HW}_N$. 
	
	For any Hermitian Pauli matrix $E\left(\bm{c},\bm{d}\right)$ and $\nu\in\{\pm 1\}$, the operator $\frac{I_N+\nu E\left(\bm{c},\bm{d}\right)}{2}$ projects onto the $\nu$-eigenspace of $E\left(\bm{c},\bm{d}\right)$. Thus, the projector onto the codespace $\mathcal{V}(\mathcal{S})$ of the stabilizer code defined by $\mathcal{S}=\langle \nu_i E\left(\bm{c_i},\bm{d_i}\right): i=1,2,\dots, r \rangle$ is 
	\begin{equation}
		\Pi_{\mathcal{S}}=\prod_{i=1}^{r}\frac{\left(I_N+\nu_i E\left(\bm{c_i},\bm{d_i}\right)\right)}{2}=\frac{1}{2^r}\sum_{j=1}^{2^r}\epsilon_j E\left(\bm{a_j},\bm{b_j}\right),
	\end{equation}
	where $\epsilon_j\in \{\pm 1 \}$ is a character of the group $\mathcal{S}$, and is determined by the signs of the generators that produce $E(\bm{a_j},\bm{b_j})$: $\epsilon_jE\left(\bm{a_j},\bm{b_j}\right)=\prod_{t\in J\subset \{1,2,\dots,r\} } \nu_t E\left(\bm{c_t},\bm{d_t}\right)$ for a unique $J$.
	
	A \textit{CSS (Calderbank-Shor-Steane) code} is a particular type of stabilizer code with generators that can be separated into strictly $X$-type and strictly $Z$-type operators. Consider two classical binary codes $\mathcal{C}_1,\mathcal{C}_2$ such that $\mathcal{C}_2\subset \mathcal{C}_1$, and let $\mathcal{C}_1^\perp$, $\mathcal{C}_2^\perp$ denote the dual codes. Note that $\mathcal{C}_1^\perp\subset \mathcal{C}_2^\perp$. Suppose that $\mathcal{C}_2 = \langle \bm{c_1},\bm{c_2},\dots,\bm{c_{k_2}} \rangle$ is an $[n,k_2]$ code and $\mathcal{C}_1^\perp =\langle \bm{d_1},\bm{d_2}\dots,\bm{d_{n-k_1}}\rangle$ is an $[n,n-k_1]$ code.
	Then, the corresponding CSS code has the stabilizer group 
	\begin{align*}
		\mathcal{S} 
		&=\langle \nu_{(\bm{c_i},\bm{0})} E\left(\bm{c_i},\bm{0}\right), \nu_{(\bm{0},\bm{d_j})} E\left(\bm{0},\bm{d_j}\right)\rangle_{i=1;~j=1}^{i=k_2;~j=n-k_1} \nonumber\\
		&=\{\epsilon_{(\bm{a},\bm{0})} \epsilon_{(\bm{0},
			\bm{b})} E\left(\bm{a},\bm{0}\right)E\left(\bm{0},\bm{b}\right): \bm{a}\in \mathcal{C}_2, \bm{b}\in \mathcal{C}_1^\perp\}, 
	\end{align*}
	where $\nu_{(\bm{c_i},\bm{0})},\nu_{(\bm{0},\bm{d_j})},\epsilon_{(\bm{a},\bm{0})},\epsilon_{(\bm{0},
		\bm{b})}  \in\{\pm 1 \}$.
	We capture sign information through character vectors 
	$\bm{y}\in \F_2^n/\mathcal{C}_1,\bm{r} \in \F_2^n/\mathcal{C}_2^\perp$ such that for any $ \epsilon_{(\bm{a},\bm{0})} \epsilon_{(\bm{0},\bm{b})} E\left(\bm{a},\bm{0}\right)E\left(\bm{0},\bm{b}\right) \in S$, we have $ \epsilon_{(\bm{a},\bm{0})} = (-1)^{\bm{a}\bm{r}^T}$ and $ \epsilon_{(\bm{0},\bm{b})} = (-1)^{\bm{b}\bm{y}^T}$. 
	If $\mathcal{C}_1$ and $\mathcal{C}_2^\perp$ can correct up to $t$ errors, then $S$ defines an $\left[\left[n, k_1-k_2, d\right]\right]$ CSS code with $d\ge 2t+1$, which we will represent as CSS($X,\mathcal{C}_2,\bm{r};Z,\mathcal{C}_1^\perp, \bm{y}$). If $G_2$ and $G_1^\perp$ are the generator matrices for $\mathcal{C}_2$ and $\mathcal{C}_1^\perp$ respectively, then the $(n-k_1+k_2)\times (2n)$ matrix
	\begin{equation}
		G_{\mathcal{S}}=\left[\begin{array}{c|c}
			G_2 &  \\ \hline
			&  G_1^\perp
		\end{array} \right]
	\end{equation}
	generates $\mathcal{S}$. 

	Since we consider diagonal gates, the signs of $X$-stabilizers do not matter. We then assume $\bm{r}=\bm{0}$ in the rest of this paper. 
	
	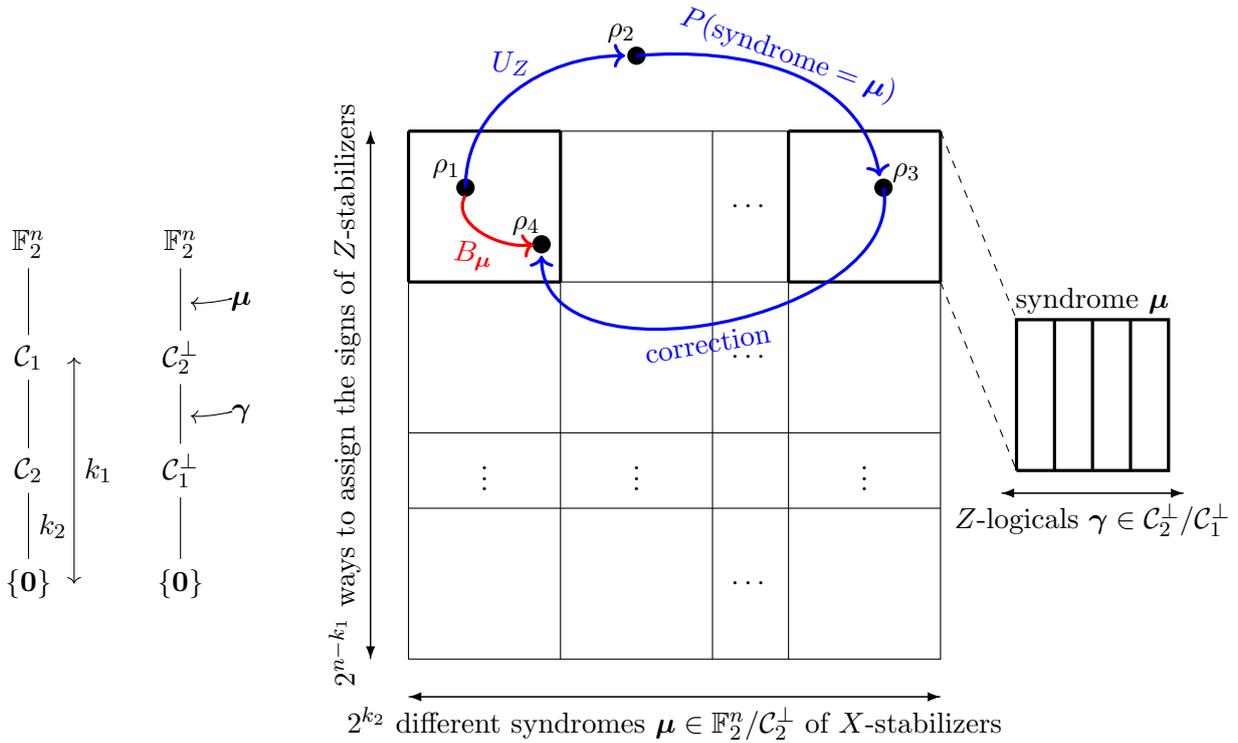
\begin{figure*}[h!]
		\centering
		\begin{tikzpicture}
			\node (Z) at (0,0) {$\{ \bm{0} \}$};
			\node (C2) at (0,1.5) {$\mathcal{C}_2$};
			\node (C1) at (0,3) {$\mathcal{C}_1$};
			\node (F2m) at (0,4.5) {$\mathbb{F}_2^{n}$};
			\path[draw] (Z) -- (C2) node[midway,right] {{$k_2$}} -- (C1) node[midway,left] {{}} -- (F2m) node[midway,left] {{}};
			\path[draw,<->,black] (0.6,0) -- (0.6,3) node [midway,right] {$k_1$}; 
			\node (Zp) at (2,0) {$\{ \bm{0} \}$};
			\node (C1p) at (2,1.5) {$\mathcal{C}_1^{\perp}$};
			\node (C2p) at (2,3) {$\mathcal{C}_2^{\perp}$};
			\node (F2m) at (2,4.5) {$\mathbb{F}_2^{n}$};
			\node (mu) at (2.8,3.75) {$\bm{\mu}$};
			\node (frommu) at (2.5,3.75) {};
			\node (tomu) at (2,3.75) {};
			\node (gamma) at (2.8,2.25) {$\bm{\gamma}$};
			\node (fromgamma) at (2.5,2.25) {};
			\node (togamma) at (2,2.25) {};
			\draw[->,black] (frommu) to [out=10,in=-10] (tomu);
			\draw[->,black] (fromgamma) to [out=10,in=-10] (togamma);
			
			\path[draw] (Zp) -- (C1p) node[midway,left] {{}} -- (C2p) node[midway,left] {{}} -- (F2m) node[midway,left] {{}};
			\tikzmath{\x1 = 5; \y1 =-1; 
				\x2 = \x1 + 7; \y2 =\y1 +7; } 
			
			\draw[] (\x1,\y1) -- (\x2,\y1) -- (\x2,\y2) -- (\x1,\y2) -- (\x1,\y1);
			\draw[] (\x1,\y1+2) -- (\x2,\y1+2);
			\draw[] (\x1,\y1+3) -- (\x2,\y1+3);
			\draw[] (\x1,\y1+5) -- (\x2,\y1+5);
			\draw[] (\x1+2,\y1) -- (\x1+2,\y2);
			\draw[] (\x1+4,\y1) -- (\x1+4,\y2);
			\draw[] (\x1+5,\y1) -- (\x1+5,\y2);
			\draw[-Latex] (\x1-0.5,\y1) -- (\x1-0.5,\y2) node [midway,left] {\rotatebox{90}{$2^{n-k_1}$ ways to assign the signs of $Z$-stabilizers}}; 
			\draw[Latex-] (\x1-0.5,\y1) -- (\x1-0.5,\y2) node [midway,left] {};
			
			\draw[-Latex] (\x1,\y1-0.5) -- (\x2,\y1-0.5) node [midway,below] {$2^{k_2}$ different syndromes $\bm{\mu} \in \F_2^n/\mathcal{C}_2^\perp$ of $X$-stabilizers}; 
			\draw[Latex-] (\x1,\y1-0.5) -- (\x2,\y1-0.5) node [midway,below] {}; 
			
			\node (vdots1) at (\x1+1,\y1+2.5) {$\vdots$};
			\node (vdots2) at (\x1+3,\y1+2.5) {$\vdots$};
			\node (vdots2) at (\x1+6,\y1+2.5) {$\vdots$};
			\node (cdots1) at (\x2-2.5,\y1+1) {$\cdots$};
			\node (cdots2) at (\x2-2.5,\y1+4) {$\cdots$};
			\node (cdots2) at (\x2-2.5,\y1+6) {$\cdots$};
			\draw[very thick] (\x1,\y2) -- (\x1+2,\y2);
			\draw[very thick] (\x1,\y2) -- (\x1,\y2-2);
			\draw[very thick] (\x1+2,\y2-2) -- (\x1+2,\y2);
			\draw[very thick] (\x1+2,\y2-2) -- (\x1,\y2-2);
			
			\draw[very thick] (\x2,\y2) -- (\x2-2,\y2);
			\draw[very thick] (\x2,\y2) -- (\x2,\y2-2);
			\draw[very thick] (\x2-2,\y2-2) -- (\x2-2,\y2);
			\draw[very thick] (\x2-2,\y2-2) -- (\x2,\y2-2);
			\filldraw[very thick](\x1+0.75,\y2-0.75) circle (0.1);
			\node (rho1) at (\x1+0.5,\y2-0.5) {$\rho_1$};
			
			\filldraw[very thick](\x1+3,\y2+1) circle (0.1);
			\node (rho2) at (\x1+2.8,\y2+1.3) {$\rho_2$};
			
			\filldraw[very thick](\x2-0.75,\y2-0.75) circle (0.1);
			\node (rho3) at (\x2-0.45,\y2-0.55) {$\rho_3$};
			
			\filldraw[very thick](\x1+1.75,\y2-1.5) circle (0.1);
			\node (rho4) at (\x1+1.55,\y2-1.25) {$\rho_4$};
			
			\draw[->,blue,very thick] (\x1+0.75,\y2-0.75) to [out=90,in=180] (\x1+2.85,\y2+1);
			\node[blue] (U_Z) at (\x1+1.35,\y2+0.9) {$U_Z$};
			
			\draw[->,blue,very thick] (\x1+3,\y2+1) to [out=5,in=100] (\x2-0.8,\y2-0.6);
			\node[blue] (msmt) at (\x2-2,\y2+1) {\rotatebox{-20}{$P(\text{syndrome}=\bm{\mu})$}};
			
			\draw[->,blue,very thick] (\x2-0.75,\y2-0.75) to [out=-80,in=-90] (\x1+1.75,\y2-1.65);
			\node[blue] (msmt) at (\x2-3,\y2-2.8) {\rotatebox{10}{correction 
			}};
			
			\draw[->,red,very thick] (\x1+0.75,\y2-0.85) to [out=-120,in=-170] (\x1+1.65,\y2-1.5);
			\node[red] (ILO) at (\x1+0.85,\y2-1.65) {\rotatebox{-5}{$B_{\bm{\mu}}$}};
			\draw[very thick] (\x2+1,\y2-4.5) -- (\x2+3,\y2-4.5) -- (\x2+3,\y2-2.5) -- (\x2+1,\y2-2.5) -- (\x2+1,\y2-4.5);
			
			\draw[dashed] (\x2,\y2) to (\x2+1,\y2-2.5);
			\draw[dashed] (\x2,\y2-2) to (\x2+1,\y2-4.5);
			
			\draw[very thick] (\x2+1.5,\y2-4.5) -- (\x2+1.5,\y2-2.5);
			\draw[very thick] (\x2+2,\y2-4.5) -- (\x2+2,\y2-2.5);
			\draw[very thick] (\x2+2.5,\y2-4.5) -- (\x2+2.5,\y2-2.5);
			\node (mu) at (\x2+2,\y2-2.25) {syndrome $\bm{\mu}$};
			
			\draw[-Latex] (\x2+3.2,\y2-4.8) -- (\x2+0.8,\y2-4.8) node [midway,below] {$Z$-logicals $\bm{\gamma}\in \mathcal{C}_2^\perp /\mathcal{C}_1^\perp$}; 
			\draw[Latex-] (\x2+3.2,\y2-4.8) -- (\x2+0.8,\y2-4.8) node [midway,below] {}; 
			
		\end{tikzpicture}
		\caption{
			The $2^{n-k_1}$ rows of the array are indexed by the $\llbr n,k_1-k_2,d \rrbr $ CSS codes corresponding to all possible signs of the $Z$-stabilizer group. The $2^{k_2}$ columns of the array are indexed by all possible $X$-syndromes $\bm{\mu}$. The logical operator $B_{\bm{\mu}}$ is induced by (1) preparing any code state $\rho_1$; (2) applying a diagonal physical gate $U_Z$ to obtain $\rho_2$; (3) using $X$-stabilizers to measure $\rho_2$, obtaining the syndrome $\bm{\mu}$ with probability $p_{\bm{\mu}}$, and the post-measurement state $\rho_3$; (4) applying a Pauli correction to $\rho_3$, obtaining $\rho_4$. Graph from \cite{generator_coeff_framework}. 
		}
		\label{fig:GC_frame} 
	\end{figure*}
	
	\subsection{Generator Coefficient Framework}
	\label{subsec:perlim_GCs}
	The \emph{Generator Coefficient Framework} was introduced in \cite{generator_coeff_framework} to describe the evolution of stabilizer code states under a physical diagonal gate $U_Z = \sum_{\bm{u}\in \F_2^n} d_{\bm{u}}\ket{\bm{u}}\bra{\bm{u}}$. Note that $\ket{\bm{u}}\bra{\bm{u}} = \frac{1}{2^n}\sum_{\bm{v}\in\F_2^n} (-1)^{\bm{uv}^T} E(\bm{0},\bm{v})$. Then we may expand $U_Z$ in the Pauli basis 
	\begin{align}
		U_Z = \sum_{\bm{v}\in\F_2^n} f(\bm{v}) E(\bm{0},\bm{v}),
	\end{align}
	where 
	\begin{align}\label{eqn:coeff_UZ}
		f(\bm{v}) = \frac{1}{2^n} \sum_{\bm{u}\in\F_2^n} (-1)^{\bm{uv}^T} d_{\bm{u}}. 
	\end{align}
	Note that we can connect the coefficients in standard basis and Pauli basis as
	\begin{align}
		[f(\bm{v})]_{\bm{v}\in\F_2^n} =  [d_{\bm{u}}]_{\bm{u}\in\F_2^n} H_{2^n}, 
	\end{align}
	where $H_{2^n} = H \otimes H_{2^{n-1}} = H^{\otimes n}$ is the Walsh-Hadamard matrix. 
	
	We consider the average logical channel induced by $U_Z$ of an $\llbr n,k,d\rrbr$ CSS($X,\mathcal{C}_2;Z,\mathcal{C}_1^\perp, \bm{y}$) code as described in Figure \ref{fig:GC_frame}. 
	Let $B_{\bm{\mu}}$ be the induced logical operator corresponding to the syndrome $\bm{\mu}$. Then the evolution of code states can be described as 
	\begin{align}
		\rho_4= \sum_{\bm{\mu}\in \F_2^n/\mathcal{C}_2^\perp} B_{\bm{\mu}} \rho_1 B_{\bm{\mu}}^\dagger. 
	\end{align}
	The generator coefficients $A_{\bm{\mu},\bm{\gamma}}$ are obtained by expanding the logical operator $B_{\bm{\mu}}$ in terms of $Z$-logical Pauli operators $\epsilon_{(\bm{0},\bm{\gamma})}E(\bm{0},\bm{\gamma})$,
	\begin{align}\label{eqn:kraus_ops}
		B_{\bm{\mu}}=\epsilon_{(\bm{0},\bm{\gamma_{\mu}})}E(\bm{0},\bm{\gamma_{\mu}}) \sum_{\bm{\gamma}\in \mathcal{C}_2^\perp / \mathcal{C}_1^\perp} A_{\bm{\mu},\bm{\gamma}}~ \epsilon_{(\bm{0},\bm{\gamma})}E(\bm{0},\bm{\gamma}),
	\end{align}
	where $\epsilon_{(\bm{0},\bm{\gamma_{\mu}})}E(\bm{0},\bm{\gamma_{\mu}})$ models the $Z$-logical Pauli operator introduced by a decoder. 
	For each pair of a $X$-syndrome $\bm{\mu} \in \F_2^n / \mathcal{C}_2^\perp$ and a $Z$-logical $\bm{\gamma} \in \mathcal{C}_2^\perp / \mathcal{C}_1^\perp$, the generator coefficient $A_{\bm{\mu},\bm{\gamma}}$ corresponding to $U_Z$ is
	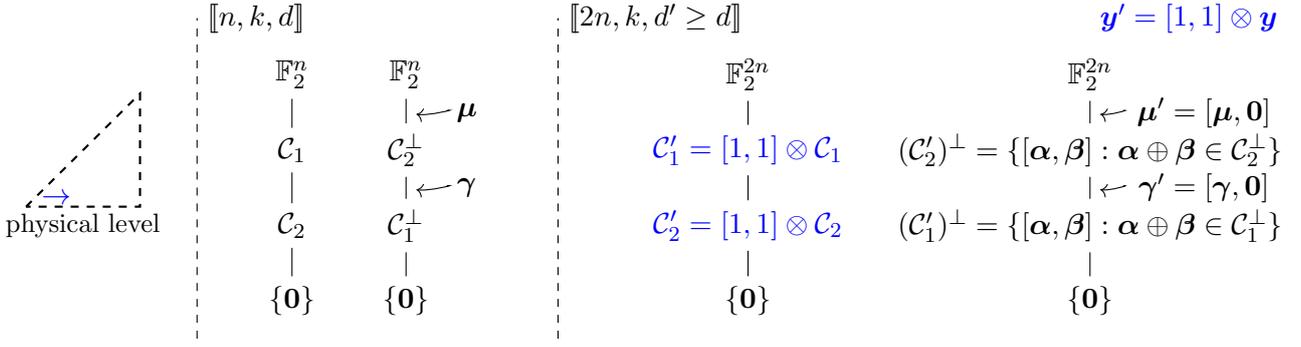
\begin{figure*}
		\centering
		\begin{tikzpicture}
			\node (Z) at (0,0) {$\{ \bm{0} \}$};
			\node (C2) at (0,1) {$\mathcal{C}_2$};
			\node (C1) at (0,2) {$\mathcal{C}_1$};
			\node (F2m) at (0,3) {$\mathbb{F}_2^{n}$};
			
			\path[draw] (Z) -- (C2)  -- (C1) -- (F2m);
			
			\node (Zp) at (1.5,0) {$\{ \bm{0} \}$};
			\node (C1p) at (1.5,1) {$\mathcal{C}_1^{\perp}$};
			\node (C2p) at (1.5,2) {$\mathcal{C}_2^{\perp}$};
			\node (F2m) at (1.5,3) {$\mathbb{F}_2^{n}$};
			\path[draw] (Zp) -- (C1p) -- (C2p)  -- (F2m); 
			\node (mu) at (2.3,2.5) {$\bm{\mu}$};
			\node (frommu) at (1.9,2.5) {};
			\node (tomu) at (1.5,2.5) {};
			\node (gamma) at (2.3,1.5) {$\bm{\gamma}$};
			\node (fromgamma) at (1.9,1.5) {};
			\node (togamma) at (1.5,1.5) {};
			\draw[->,black] (frommu) to [out=10,in=-10] (tomu);
			\draw[->,black] (fromgamma) to [out=10,in=-10] (togamma);
			\node (Z) at (6,0) {$\{ \bm{0} \}$};
			\node (C2) at (6,1) {$\color{blue}\mathcal{C}'_2=[1,1]\otimes \mathcal{C}_2$};
			\node (C1) at (6,2) {$\color{blue}\mathcal{C}'_1=[1,1]\otimes \mathcal{C}_1$};
			\node (F2m) at (6,3) {$\mathbb{F}_2^{2n}$};
			
			\path[draw] (Z) -- (C2)  -- (C1)  -- (F2m); 
			
			\node (Zp) at (10.5,0) {$\{ \bm{0} \}$};
			\node (C1p) at (10.5,1) {$(\mathcal{C}'_1)^\perp=\{[\bm{\alpha},\bm{\beta}]: \bm{\alpha}\oplus\bm{\beta} \in \mathcal{C}_1^\perp \}$};
			\node (C2p) at (10.5,2) {$(\mathcal{C}'_2)^\perp=\{[\bm{\alpha},\bm{\beta}]: \bm{\alpha}\oplus\bm{\beta} \in \mathcal{C}_2^\perp \}$};
			\node (F2m) at (10.5,3) {$\mathbb{F}_2^{2n}$};
			
			\path[draw] (Zp) -- (C1p) -- (C2p) -- (F2m); 
			
			\draw [dashed,black] (3.5,-0.5) -- (3.5,3.75)
			node[at end,right] {{$\llbr 2n,k,d'\ge d \rrbr$}};
			\draw [dashed,black] (-1.25,-0.5) -- (-1.25,3.75)
			node[at end,right] {{$\llbr n,k,d \rrbr$}};
			\node (y) at (11.8,3.75) {$\color{blue}\bm{y}' = [1,1]\otimes \bm{y}$};
			
			\node (mu_n) at (12,2.5) {$\bm{\mu}'=[\bm{\mu},\bm{0}]$};
			\node (frommu_n) at (10.8,2.5) {};
			\node (tomu_n) at (10.5,2.5) {};
			\node (gamma_n) at (12,1.5) {$\bm{\gamma}'=\left[\bm{\gamma},\bm{0}\right]$};
			\node (fromgamma_n) at (10.8,1.5) {};
			\node (togamma_n) at (10.5,1.5) {};
			\draw[->,black] (frommu_n) to [out=10,in=-10] (tomu_n);
			\draw[->,black] (fromgamma_n) to [out=10,in=-10] (togamma_n);
			
			\draw[thick,dashed] (-3.5,1.25) -- (-2,2.75) -- (-2,1.25)  -- (-3.5,1.25);
			\node[blue] (phys) at (-3.1,1.35) {$\rightarrow$};
			\node at (-2.75,1) {\small physical level};
		\end{tikzpicture}
		\caption{Concatenation transforms an $\llbr n,k,d \rrbr$ CSS code preserved by a diagonal gate $U_Z$ at level $l$ to a $\llbr 2n,k,d'\rrbr$ CSS code preserved by a family of diagonal gates $U'_Z$, some of which are at level $l+1$. The logical operator induced by $U'_Z$ coincides with the logical operator induced by $U_Z$. } 
		\label{fig:concatenation}
	\end{figure*}
	\begin{equation}\label{eqn:def_GC_UZ}
		A_{\bm{\mu},\bm{\gamma}}\coloneqq \sum_{\bm{z}\in \mathcal{C}_1^\perp+\bm{\mu}+\bm{\gamma}}\epsilon_{(\bm{0},\bm{z})}f(\bm{z}),
	\end{equation}
	where $\epsilon_{(\bm{0},\bm{z})} = (-1)^{\bm{z}\bm{y}^T}$ is the sign of the $Z$-stabilizer $E(\bm{0},\bm{z})$. Based on the structure of CSS codes, generator coefficients group the Pauli coefficients of $U_Z$ together, and tune them by the signs of $Z$-stabilizer. 
	We use \eqref{eqn:coeff_UZ} to simplify \eqref{eqn:def_GC_UZ} as
	\begin{align}\label{eqn:GC_C1_y}
		A_{\bm{\mu},\bm{\gamma}}
		&= \frac{1}{2^n} \sum_{\bm{u}\in\F_2^n} \sum_{\bm{z}\in \mathcal{C}_1^\perp+\bm{\mu}+\bm{\gamma}} (-1)^{\bm{zy}^T} (-1)^{\bm{zu}^T} d_{\bm{u}} \nonumber\\
		&=\frac{1}{|\mathcal{C}_1|}\sum_{\bm{u}\in \mathcal{C}_1+\bm{y}} (-1)^{(\bm{\mu}\oplus\bm{\gamma})(\bm{y}\oplus\bm{u})^T}  d_{\bm{u}}.
	\end{align}
	
	The diagonal physical gate $U_Z$ preserves a CSS($X,\mathcal{C}_2;Z,\mathcal{C}_1^\perp, \bm{y}$) codespace if and only if $B_{\bm{\mu}=\bm{0}}$ is a unitary \cite{generator_coeff_framework}, which is equivalent to requiring
	\begin{equation}\label{eqn:preserved_by_Uz}
		\sum_{\bm{\gamma}\in \mathcal{C}_2^\perp/\mathcal{C}_1^\perp} |A_{\bm{0},\bm{\gamma}}|^2=1.
	\end{equation}
	Note that \eqref{eqn:preserved_by_Uz} is also equivalent to $A_{\bm{\mu}\neq \bm{0},\bm{\gamma}} = 0 $ for all $\bm{\gamma}\in \mathcal{C}_2^\perp/\mathcal{C}_1^\perp$.
	The induced logical operator is 
	\begin{equation}\label{eqn:dig_log_op}
		U_Z^L 
		= \sum_{\bm{\alpha}\in \F_2^k} A_{\bm{0}, g(\bm{\alpha})} E(\bm{0},\bm{\alpha}),
	\end{equation}
	where $g: \F_2^k \to \mathcal{C}_2^\perp / \mathcal{C}_1^\perp $ is a bijective map defined by $g(\bm{\alpha}) =\bm{\alpha} G_{\mathcal{C}_2^\perp /\mathcal{C}_1^\perp}$, and $G_{\mathcal{C}_2^\perp /\mathcal{C}_1^\perp}$ is a generator matrix of the $Z$-logicals $\mathcal{C}_2^\perp /\mathcal{C}_1^\perp$. 
	
	\section{Climbing the Physical Hierarchy}
	\label{sec:concatenation}
	We need to climb the physical Clifford hierarchy because the level of the physical operator bounds that of the induced logical operator. Consider a physical diagonal gate 
	\begin{align}
		U_Z = \sum_{\bm{u}\in \F_2^n} d_{\bm{u}}\ket{\bm{u}}\bra{\bm{u}},
	\end{align}
	that preserves an $\llbr n,k,d \rrbr$ CSS($X,\mathcal{C}_2;Z,\mathcal{C}_1^\perp, \bm{y}$) code with $X$-distance
	\begin{align}
		d_X \coloneqq \min_{\bm{x}\in \mathcal{C}_1\setminus \mathcal{C}_2}w_H(\bm{x}),
	\end{align}
	and $Z$-distance 
	\begin{align}
		d_Z \coloneqq \min_{\bm{z}\in \mathcal{C}_2^\perp \setminus \mathcal{C}_1^\perp}w_H(\bm{z}).
	\end{align}
	We denote the logical operator induced by $U_Z$ as $U_Z^L$. The concatenation process described in Figure \ref{fig:concatenation} produces a 
	$\llbr 2n,k,d' \rrbr$ CSS($X,\mathcal{C}'_2;Z,(\mathcal{C}_1')^\perp, \bm{y}'$) code. Concatenation does not change the number of $Z$-logicals or the number of $X$-syndromes, and so the number of generator coefficients remains the same. We now show this code is preserved by an ensemble of physical gates, all inducing the same logical operator as $U^L_Z$. 
	\begin{theorem}
		\label{thm:concatenation}
		The $\llbr 2n,k,d' \rrbr$ CSS($X,\mathcal{C}'_2;Z,(\mathcal{C}_1')^\perp, \bm{y}'$) code is preserved by any diagonal physical gate
		\begin{align}
			U'_Z = \sum_{\bm{u}'\in\F_2^{2n}}d'_{\bm{u}'}\ket{\bm{u}'}\bra{\bm{u}'},
		\end{align} 
		for which $d'_{[\bm{u},\bm{u}]}=d_{\bm{u}}$ for all $\bm{u}\in\F_2^n$. 
		
		\noindent The minimum distance $d'\ge d$ and the induced logical operator $(U_Z')^L$ is equal to $U_Z^L$.
	\end{theorem}
	\begin{proof}
		Let $d'_X,d'_Z$ be the $X$- and $Z$-distances for the CSS($X,\mathcal{C}'_2;Z,(\mathcal{C}_1')^\perp, \bm{y}'$) code. Given $\bm{x}'\in \mathcal{C}'_1\setminus \mathcal{C}'_2$, there exists $\bm{x}\in \mathcal{C}_1\setminus \mathcal{C}_2$ such that $\bm{x}'=[1,1]\otimes \bm{x}$, and so $d'_X=2d_X$. Given $[\bm{\alpha},\bm{\beta}]\in (\mathcal{C}'_2)^{\perp}\setminus (\mathcal{C}'_1)^{\perp}$, we have 
		\begin{align}
			w_H([\bm{\alpha},\bm{\beta}])=w_H(\bm{\alpha})+ w_H(\bm{\beta})\ge w_H(\bm{\alpha} \oplus \bm{\beta} ), 
		\end{align}
		and so $d_Z' \ge d_Z$. Concatenation doubles $X$-distance while maintaining $Z$-distance.  
		
		We now prove that $(U'_Z)^L=U_Z^L$ by showing the generator coefficients remain the same:
		\begin{align*}
			A'_{\bm{\mu}',\bm{\gamma}'}(U'_Z) 
			&= \frac{1}{|\mathcal{C}'_1|} \sum_{\bm{u}' \in \mathcal{C}'_1 +\bm{y}'} (-1)^{(\bm{\mu}'\oplus \bm{\gamma}')(\bm{y}' \oplus \bm{u}')^T} d'_{\bm{u}'} \\
			&= \frac{1}{|\mathcal{C}_1|} \sum_{\bm{u} \in \mathcal{C}_1 +\bm{y}} (-1)^{[\bm{\mu}\oplus\bm{\gamma},\bm{0}][\bm{y}\oplus\bm{u},\bm{y}\oplus\bm{u}]^T} d'_{[\bm{u},\bm{u}]} \\
			&= \frac{1}{|\mathcal{C}_1|} \sum_{\bm{u}\in \mathcal{C}_1+\bm{y}} (-1)^{(\bm{\mu}\oplus \bm{\gamma})(\bm{y}\oplus \bm{u})^T} d_{\bm{u}}\\
			&= A_{\bm{\mu},\bm{\gamma}}(U_Z).\numberthis
		\end{align*}
		Hence, concatenation brings more freedom of physical operators to realize the same logical operator.  
	\end{proof}
	We may partition $U'_Z$ into $2^n$ blocks, where the block indexed by $\bm{u}\in\F_2^n$ is a $2^n\times 2^n$ diagonal matrix diag$[d'_{[\bm{u},\bm{v}]}]$. Theorem \ref{thm:concatenation} specifies a single diagonal entry $d'_{[\bm{u},\bm{u}]}$ in each block. The remaining $2^{2n}-2^n$ entries can be freely chosen to design the unitary $U'_Z$. When $U_Z$ (on $n$ qubits) is a transversal $\mathrm{C}^{(i)}Z^{1/2^j}$ gate at level $i+j$ in the clifford hierarchy, we can choose $U_Z'$ to be the transversal $\mathrm{C}^{(i)}Z^{1/2^{j+1}}$ gate (on $2n$ qubits) at level $i+j+1$.
	\begin{remark}[Quadratic Form Diagonal (QFD) gates] 
		\label{rem:QFD_concat}
		\normalfont
		We now describe how to raise the level of a QFD gate $\tau_R^{(l)}\in \mathcal{C}_d^{(l)}$ at level $l$ in the Clifford hierarchy. Here
		\begin{equation}\label{eqn:qfd_def}
			\tau_R^{(l)} = \sum_{\bm{v}\in \F_2^n} \xi_l^{\bm{v}R\bm{v}^T \bmod {2^l}} \ket{\bm{v}}\bra{\bm{v}},
		\end{equation}
		where $\xi_l=e^{\imath \frac{\pi}{2^{l-1}}}$, and $R$ is an $n\times n$ symmetric matrix with entries in $\Z_{2^l}$, the ring of integers modulo $2^l$. Note that the exponent $vRv^T \in \Z_{2^l}$. 
		Rengaswamy et al. \cite{rengaswamy2019unifying} proved that QFD gates include all 1-local and 2-local diagonal gates in the Clifford hierarchy. We choose $U_Z' = \tau_{I_2\otimes R}^{(l+1)} \in \mathcal{C}_d^{(l+1)}$, and observe
		\begin{align}
			d'_{\bm{u},\bm{u}} = \xi_{l+1}^{2\bm{u}R\bm{u}^T} = \xi_{l}^{\bm{u}R\bm{u}^T}=d_{\bm{u}}.
		\end{align}
	\end{remark}
	
	\begin{example}[Climbing from $P^{\otimes 7}$ acting on the $\llbr 7,1,3 \rrbr$ Steane code to $T^{\otimes 14}$ acting on the $\llbr 14,1,3 \rrbr$ CSS code]
		\label{examp:concat_steane}
		\normalfont
		The Steane code \cite{steane1996multiple} is a CSS($X,\mathcal{C}_2;Z,\mathcal{C}_1^\perp,\bm{y}=\bm{0}$)  code with generator matrix
		\begin{align}
			\setlength\aboverulesep{0pt}\setlength\belowrulesep{0pt}
			\setlength\cmidrulewidth{0.5pt}
			G_{\mathcal{S}} = 
			\left[
			\begin{array}{c|c}
				H &  \\
				\hline
				& H\\
			\end{array}
			\right],
		\end{align}
		where $H$ is the parity-check matrix of the Hamming code:
		\begin{align}
			\setlength\aboverulesep{0pt}\setlength\belowrulesep{0pt}
			\setlength\cmidrulewidth{0.5pt}
			H = 
			\left[
			\begin{array}{ccccccc}
				1 & 1 & 1 & 1 & 0 & 0 & 0    \\
				1 & 1 & 0 & 0 & 1 & 1 & 0   \\
				1 & 0 & 1 & 0 & 1 & 0 & 1   \\
			\end{array}
			\right].
		\end{align}
		The only nontrivial $Z$-logical corresponds to the all one vector $\bm{1}$. After concatenation described in Figure \ref{fig:concatenation}, we obtain a $\llbr 14,1,3\rrbr$ CSS code.
		When $R=I_n$, $\tau_R^{(2)} = P^{\otimes n}$ and $\tau_{I_2\otimes R}^{(3)} = T^{\otimes 2n}$. 
		Let $A_{\bm{\mu},\bm{\gamma}}\left(\frac{\pi}{2}\right)$ and $A'_{\bm{\mu}',\bm{\gamma}'}\left(\frac{\pi}{4}\right)$ be the generator coefficients corresponding to $P^{\otimes 7}$ and $T^{\otimes 14}$ acts on the $\llbr 7,1,3 \rrbr$ and $\llbr 14,1,3\rrbr$ code respectively. Then, we have 
		\begin{align}
			A_{\bm{\mu}=\bm{0},\bm{\gamma}=\bm{0}}\left(\frac{\pi}{2}\right) 
			&= A'_{[\bm{0},\bm{0}],[\bm{1},\bm{0}]}\left(\frac{\pi}{4}\right) 
			= \cos\left(\frac{\pi}{4}\right), \nonumber \\
			A_{\bm{\mu}=\bm{0},\bm{\gamma}=\bm{1}}\left(\frac{\pi}{2}\right) 
			&= A'_{[\bm{0},\bm{0}],[\bm{1},\bm{0}]}\left(\frac{\pi}{4}\right) 
			= \imath\sin\left(\frac{\pi}{4}\right),
		\end{align}
		which implies that the invariance of $\llbr 7,1,3\rrbr$ under $P^{\otimes 7}$ and that of $\llbr 14,1,3 \rrbr$ under $T^{\otimes 14}$. It then follows from the expression of the induced logical operator in \eqref{eqn:dig_log_op} that both of the codes implement a logical $P^\dagger$. 
	\end{example}
	
	\begin{example}[Climbing from C$Z^{\otimes 2}$ acting on the $\llbr 4,2,2 \rrbr$ CSS code to C$P^{\otimes 4}$ acting on the $\llbr 8,2,2 \rrbr$ CSS code]
		\label{examp:concat_4,2,2}
		\normalfont
		Consider the $\llbr 4,2,2 \rrbr$ CSS($X,\mathcal{C}_2 ;Z,\mathcal{C}_1^\perp$) code with $\mathcal{C}_2 = \mathcal{C}_1^\perp = \{\bm{0},\bm{1}\}$.
		We may choose the generators of $Z$-logicals to be $\bm{\gamma_1} = [0,0,1,1]$ and $\bm{\gamma_2} = [0,1,1,0]$. Their generator coefficients coincide:
		\begin{align}
			&A_{\bm{\mu}=\bm{0},\bm{\gamma}=\bm{0}}(\mathrm{C}Z^{\otimes 2}) = 
			A'_{[\bm{0},\bm{0}],[\bm{0},\bm{0}]}(\mathrm{C}P^{\otimes 4}) = \frac{1}{2}, \nonumber \\
			&A_{\bm{\mu}=\bm{0},\bm{\gamma}=\bm{\gamma_1}}(\mathrm{C}Z^{\otimes 2}) =
			A'_{[\bm{0},\bm{0}],[\bm{\gamma_1},\bm{0}]}(\mathrm{C}P^{\otimes 4}) = -\frac{1}{2}, \nonumber \\
			& A_{\bm{\mu}=\bm{0},\bm{\gamma}=\bm{\gamma_2}}(\mathrm{C}Z^{\otimes 2})=
			A'_{[\bm{0},\bm{0}],[\bm{\gamma_2},\bm{0}]}(\mathrm{C}P^{\otimes 4}) = \frac{1}{2}, \nonumber \\
			&A_{\bm{\mu}=\bm{0},\bm{\gamma}=\bm{\gamma_1}\oplus \bm{\gamma_2}}(\mathrm{C}Z^{\otimes 2}) = 
			A'_{[\bm{0},\bm{0}],[\bm{\gamma_1}\oplus \bm{\gamma_2},\bm{0}]}(\mathrm{C}P^{\otimes 4}) = \frac{1}{2}.
		\end{align}
		Both cases realize a logical $Z_1 \circ$C$Z\coloneqq(Z\otimes I)\mathrm{C}Z$.
	\end{example}
	
	\begin{figure*}[!ht]
		\begin{tikzpicture}
			\draw[dashed,black] (8,0) -- (8,6.3);
			\hspace{-275pt}
			\node (b) at (12.75,6) {(a) Removing a non-trivial $Z$-stabilizer};
			\node (Z) at (12,0) {$\{ \bm{0} \}$};
			\node (C2) at (12,1.75) {$\mathcal{C}_2$};
			\node (C1) at (12,3.5) {$\mathcal{C}_1$};
			\node (F2m) at (12,5.25) {$\mathbb{F}_2^{n}$};
			\draw[->,black,dotted,thick] (11.8,4.375) to [out=180,in=180] (11.8,2.625);g
			\node (w0) at (11,3.5) {$\color{blue}\bm{w_0}$};
			\path[draw] (Z) -- (C2)  -- (C1) -- (F2m);
			
			\node (Zp) at (14.5,0) {$\{ \bm{0} \}$};
			\node (C1p) at (14.5,1.75) {$\mathcal{C}_1^{\perp}$};
			\node (C2p) at (14.5, 3.5) {$\mathcal{C}_2^{\perp}$};
			\node (F2m) at (14.5,5.25) {$\mathbb{F}_2^{n}$};
			\draw[->,black,dotted,thick] (14.7,0.875) to [out=-10,in=10] (14.7,2.625);
			\node (w0) at (15.5,1.75) {$\color{blue}\bm{\gamma_0}$};
			\path[draw] (Zp) -- (C1p) -- (C2p)  -- (F2m); 
			
			\hspace{530pt}
			\draw[] (3.5,5) -- (6.5,5) -- (6.5,3) -- (3.5,3) -- (3.5,5);
			\draw[] (3.5,4.7) -- (6.5,4.7);
			\draw[fill=red] (3.5,5) rectangle (3.8,4.7); 
			\node (gamma) at (5,5.3) {$\bm{\gamma}\in \mathcal{C}_2^\perp / \mathcal{C}_1^\perp$};
			\node (mu) at (3.3,4) {$\bm{\mu}$};
			\node (A_0,gamma) at (4.2,4.5) {\color{red}\small$A_{\bm{\mu},\bm{\gamma}}$};
			
			\draw[] (2,2) -- (8,2) -- (8,0) -- (2,0) -- (2,2);
			\draw[] (5,2) -- (5,0); 
			\draw[] (2,1.7) -- (8,1.7);
			\draw[fill=red] (2,2) rectangle (2.3,1.7); 
			\draw[fill=red] (5,2) rectangle (5.3,1.7); 
			\node (gamma') at (5,2.3) {$\bm{\gamma}'\in \langle \mathcal{C}_2^\perp / \mathcal{C}_1^\perp, {\color{blue}\bm{\gamma_0}} \rangle$};
			\node (mu') at (1.35,1) {$\bm{\mu}'=\bm{\mu}$};
			\node (A'_0,gamma) at (2.8,1.5) {\color{red}\small$A'_{\bm{\mu},\bm{\gamma}' = \bm{\gamma}}$};
			\node (A'_0,gamma) at (6.1,1.5) {\color{red}\small$A'_{\bm{\mu},\bm{\gamma}' =\bm{\gamma}\oplus \bm{\gamma_0}}$};
			
			\node (eqn) at (1.85,6) {(b) $\color{red}A_{\bm{\mu},\bm{\gamma}} = A'_{\bm{\mu},\bm{\gamma}} + A'_{\bm{\mu},\bm{\gamma}\oplus \bm{\gamma_0}}$}; 
			\draw[->] (1.25,2.5) to (1.25,3.5);
			\node (rm) at (0.25,3) {remove};
			\draw[<-] (1,2.5) to (1,3.5);
			\node (ad) at (1.65,3) {add};
		\end{tikzpicture}
		\caption{(a) Removing a $Z$-stabilizer $\bm{\gamma_0}$ creates a new $Z$-logical, and transforms an old $Z$-syndrome $\bm{w_0}$ into a new $X$-logical. (b) Removing/adding a $Z$-stabilizer induces splitting/grouping of generator coefficients.}
		\label{fig:rem_Z_stab}
	\end{figure*}
	
	\begin{remark}[Switching between Computation and Storage]
		\label{rem:sw_comp_storage}
		\normalfont
		It is the choice of character vector that distinguishes the method of concatenation depicted in Figure \ref{fig:concatenation} from the method of constructing a decoherence-free subspaces (DFS) described in \cite{coherent_noise}. Consider the graph where the vertices are the qubits involved in the support of some X-stabilizer, and where two vertices are joined by an edge if there exists a weight 2 $Z$-stabilizer involving these two qubits. Instead of choosing $\bm{y}'=[1,1]\otimes \bm{y}$, Liang, Hu et al. \cite{coherent_noise} balance the signs of $Z$-stabilizers by requiring that the support of $\bm{y}''$ include half the qubits in every connected component of the graph. The stabilizer group determines a resolution of the identity. 
		To change the signs of $Z$-stabilizers, we simply apply some physical Pauli $X$ to transform from one part of the resolution to the other part (see \cite[Example 3]{generator_coeff_framework} for more details). To determine the specific position to add these extra Pauli $X$, we consider the general encoding map 
		$g_e:\ket{\bm{\alpha}}_L\in \mathbb{F}_2^k \to \ket{\overline{\bm{\alpha}}}\in \mathcal{V}(\mathcal{S})$ of a CSS($X,\mathcal{C}_2,\bm{r};Z,\mathcal{C}_1^\perp, \bm{y}$) code \cite{generator_coeff_framework},
		\begin{equation} \label{eqn:gen_encode_map}
			\ket{\overline{\bm{\alpha}}}
			\coloneqq \frac{1}{\sqrt{|\mathcal{C}_2|}} \sum_{\bm{x}\in \mathcal{C}_2} (-1)^{\bm{x}\bm{r}^T}\ket{\bm{\alpha}  G_{\mathcal{C}_1/\mathcal{C}_2} \oplus \bm{x} \oplus \bm{y}},
		\end{equation}
		where $\bm{r},\bm{y}$ are the character vectors for $X$- and $Z$-stabilizers, and $G_{\mathcal{C}_1/\mathcal{C}_2}$ is a generator matrix of the $X$-logicals $\mathcal{C}_1/\mathcal{C}_2$. The positions of these Pauli $X$ correspond to the support of the difference of two character vectors $\bm{y}'-\bm{y}''$. 
		Hence it is simple to switch between computation and storage. Given a code that realizes a specific diagonal logical operator induced by the physical gate $U_Z$, we first apply the concatenation described in Figure \ref{fig:concatenation}. After concatenation, we choose $U'_Z = I_{N} \otimes U_Z$, at the same level as $U_Z$, to realize the same specific logical operator. We then apply some physical Pauli $X$ to change signs of $Z$-stabilizers and embed the logical information in a DFS. To continue the computation, we recover the stored results by applying the same Pauli $X$. Note that concatenation doubles the $X$-distance, which improves protection when we change the signs of $Z$-stabilizers. 
		
		For example, suppose our goal is to first implement a logical $P^\dagger$ and to wait for a while before calculating the next step. We can apply the physical $U'_Z=I_2^{\otimes 7}\otimes P^{\otimes 7}$ to the $\llbr 14,1,3\rrbr$ CSS code in Example \ref{examp:concat_steane} to realize the logical $P^\dagger$. Note that $\bm{y}'=\bm{0}\in\F_2^{14}$ and one choice of $\bm{y}''$ is $[1,0]\otimes \bm{1}_7 \in \F_2^{14}$. Then we can apply Pauli X alternatively to map the computed result in a DFS.
	\end{remark}
	To achieve more advanced computation, we need diagonal logical operators from  higher levels. Raising the level of a physical operator prepares the ground for climbing the logical hierarchy.

	\section{Climbing the Logical Hierarchy}
	\label{sec:rem_Z_stab}
	In this Section, we describe how to increase the level of an induced logical operator by judiciously removing $Z$-stabilizers from a CSS code. We start by considering a physical diagonal gate
	\begin{align}
		U_Z=\sum_{\bm{v}\in \F_2^n}f(\bm{v})E(\bm{0},\bm{v})
	\end{align}
	that preserves an $\llbr n,k,d \rrbr$ CSS$(X,\mathcal{C}_2;Z,\mathcal{C}_1^\perp,\bm{y})$ code. The induced logical channels are described by generator coefficients $A_{\bm{\mu},\bm{\gamma}}$ where $\bm{\mu} \in \F_2^n/\mathcal{C}_2^\perp$ and $\bm{\gamma}\in \mathcal{C}_2^\perp /  \mathcal{C}_1^\perp$. Let $\bm{\gamma_0}\in\mathcal{C}_1^\perp$ be a nontrivial $Z$-stabilizer. Set $\mathcal{C}_1^\perp = \langle (\mathcal{C}'_1)^\perp,\bm{\gamma_0} \rangle$, and set $\mathcal{C}'_1 = \langle \mathcal{C}_1,\bm{w_0} \rangle$, where $\bm{w_0} \in \F_2^n/\mathcal{C}_1$. If we remove $\bm{\gamma_0}$ from $\mathcal{C}_1^\perp$, then $\bm{\gamma_0}$ becomes a $Z$-logical for the $\llbr n,k+1,d'\le d \rrbr$ CSS$(X,\mathcal{C}_2;Z,(\mathcal{C}'_1)^\perp,\bm{y})$ code, as shown in Figure \ref{fig:rem_Z_stab}(a). Removing the $Z$-logical $\bm{\gamma_0}$ doubles the number of $Z$-logicals. Each generator coefficient $A_{\bm{\mu},\bm{\gamma}}$ associated with the original CSS code splits into two generator coefficients $A'_{\bm{\mu},\bm{\gamma}'=\bm{\gamma}}$ and $A'_{\bm{\mu},\bm{\gamma}'=\bm{\gamma}\oplus\bm{\gamma_0}}$ associated with the new code. We have 
	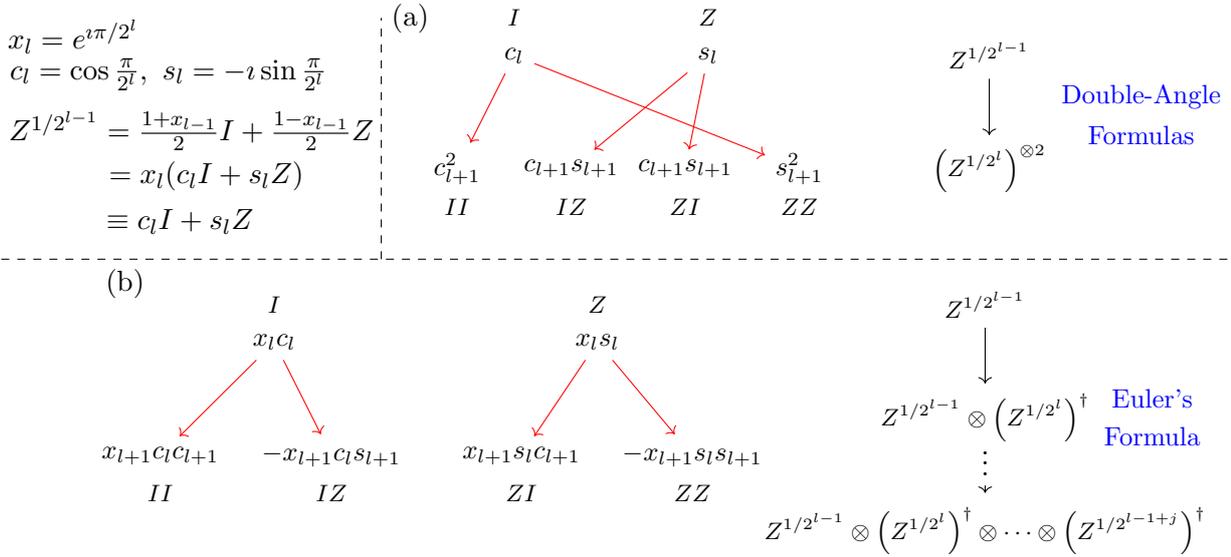
\begin{figure*}[!ht]
		\centering
		\begin{tikzpicture}
			\tikzmath{\x1 = 8.5; \y1 =-1.3; 
				\x2 = 5; \y2 =4.7; } 
			\draw [dashed,black] (3.5,7.8) -- (3.5,11)
			node[at end,right] {{(a)}};
			\draw [dashed,black] (-1.5,7.8) -- (14.8,7.8)
			node[pos=0.1,below] {{(b)}};
			
			\node () at (5.25,6+\x2) {\footnotesize $I$};
			\node () at (7.8,6+\x2) {\footnotesize $Z$};
			\node (0n) at (5.25,5.5+\x2) {\small $c_{l}$}; 
			\node (1n) at (7.8,5.5+\x2) {\small $s_{l} $}; 
			
			\node () at (4.5,3.5+\x2) {\footnotesize $II$};
			\node () at (6,3.5+\x2) {\footnotesize $IZ$};
			\node () at (7.5,3.5+\x2) {\footnotesize $ZI$};
			\node () at (9,3.5+\x2) {\footnotesize $ZZ$};
			\node (00n) at (4.5,4+\x2) {\small $c^2_{l+1}$}; 
			\node (01n) at (6,4+\x2) {\small $c_{l+1}s_{l+1}$}; 
			\node (10n) at (7.5,4+\x2) {\small $c_{l+1}s_{l+1}$};
			\node (11n) at (9,4+\x2) {\small $s^2_{l+1}$};
			
			\draw[->,red] (0n) -- (00n);
			\draw[->,red] (0n) -- (11n);
			\draw[->,red] (1n) -- (10n);
			\draw[->,red] (1n) -- (01n);
			
			\node (Z1) at (11.5,5.5+\x2) {\footnotesize $Z^{1/2^{l-1}}$};
			\node (CZ1) at (11.5,4+\x2) {\footnotesize $\left(Z^{1/2^{l}}\right)^{\otimes 2}$};
			\draw[->] (Z1) to (CZ1);
			\node () at (13.5,4.95+\x2) {\small \color{blue}Double-Angle};
			\node () at (13.5,4.45+\x2) {\small \color{blue}Formulas};
			\node () at (-0.55,2.25+\x1) {$x_l = e^{\imath \pi/2^l}$};
			\node () at (0.7,1.75+\x1) {$c_{l} = \cos\frac{\pi}{2^{l}},~s_{l} = -\imath\sin\frac{\pi}{2^{l}}$};
			\node () at (1,1+\x1) {$Z^{1/2^{l-1}} = \frac{1+x_{l-1}}{2}I+\frac{1-x_{l-1}}{2}Z$};
			\node () at (1.2,0.4+\x1) {$ = x_{l}(c_l I+s_l Z)$}; 
			\node () at (0.85,-0.2+\x1) {$\equiv c_l I+s_l Z$};
			\hspace{-40pt}
			\node () at (3.5,2.5+\y2) {\footnotesize $I$};
			\node () at (7.75,2.5+\y2) {\footnotesize $Z$};
			\node (0n) at (3.5,2+\y2) {\small $x_{l}c_{l}$}; 
			\node (1n) at (7.75,2+\y2) {\small $x_{l}s_{l}$};
			
			\node () at (2,0+\y2) {\footnotesize $II$};
			\node () at (4.25,0+\y2) {\footnotesize $IZ$};
			\node () at (6.75,0+\y2) {\footnotesize $ZI$};
			\node () at (9,0+\y2) {\footnotesize $ZZ$};
			\node (00n) at (2,0.5+\y2) {\small $x_{l+1}c_{l}c_{l+1}$}; 
			\node (01n) at (4.25,0.5+\y2) {\small $-x_{l+1}c_{l}s_{l+1}$}; 
			\node (10n) at (6.75,0.5+\y2) {\small $x_{l+1}s_{l}c_{l+1}$};
			\node (11n) at (9,0.5+\y2) {\small $-x_{l+1}s_{l}s_{l+1}$};
			
			\draw[->,red] (0n) -- (00n);
			\draw[->,red] (0n) -- (01n);
			\draw[->,red] (1n) -- (10n);
			\draw[->,red] (1n) -- (11n);
			\hspace{10pt}
			\node (Z1) at (12.5,2.5+\y2) {\footnotesize $Z^{1/2^{l-1}}$};
			\node (CZ1) at (12.5,1+\y2) {\footnotesize $Z^{1/2^{l-1}} \otimes \left(Z^{1/2^{l}}\right)^\dagger$};
			\node (CnZ1) at (12.5,-0.5+\y2) {\footnotesize $Z^{1/2^{l-1}} \otimes \left(Z^{1/2^{l}}\right)^\dagger \otimes \cdots \otimes \left(Z^{1/2^{l-1+j}}\right)^\dagger$};
			\draw[->] (Z1) to (CZ1);
			\draw[->] (12.5,0.1+\y2) to (12.5,0+\y2);
			\node (2dots) at (12.5,0.5+\y2) {\vdots};
			\node () at (14.7,1.25+\y2) {\small \color{blue} Euler's};
			\node () at (14.7,0.75+\y2) {\small \color{blue} Formula};
			
		\end{tikzpicture}
		\caption{Admissible Splits of $Z$-rotations: (a) One step for uniform rotations from $Z^{1/2^{l-1}}$ to $Z^{1/2^l}\otimes Z^{1/2^l}$; (b) Multi-step for non-uniform rotations $Z\to Z \otimes P^\dagger \to Z\otimes P^\dagger \otimes T^\dagger \to \cdots$.}
		\label{fig:Z-rot_split_schemes}
	\end{figure*}
	\begin{align*}
		A_{\bm{\mu},\bm{\gamma}} 
		&= \sum_{\bm{z} \in \langle (\mathcal{C}'_1)^\perp, \bm{\gamma_0} \rangle +\bm{\mu} +\bm{\gamma}} \epsilon_{(\bm{0},\bm{z})} f(\bm{z}) \\
		&= \sum_{\bm{z} \in (\mathcal{C}'_1)^\perp+\bm{\mu} +\bm{\gamma}} \epsilon_{(\bm{0},\bm{z})} f(\bm{z})\\
		&~~~+ \sum_{\bm{z} \in (\mathcal{C}'_1)^\perp +\bm{\mu} +\bm{\gamma}+ \bm{\gamma_0}} \epsilon_{(\bm{0},\bm{z})} f(\bm{z})\\
		&= A'_{\bm{\mu},\bm{\gamma}'=\bm{\gamma}} +A'_{\bm{\mu},\bm{\gamma}'=\bm{\gamma}\oplus\bm{\gamma_0}} .\label{eqn:group_split_GCs} \numberthis
	\end{align*}
	Adding a $Z$-stabilizer simply reverses this process as shown in Figure \ref{fig:rem_Z_stab}(b). 
	\begin{definition}[Admissible Splits]
		A split is \underline{admissible} if the physical diagonal gate $U_Z$ preserves the CSS$(X,\mathcal{C}_2;Z,(\mathcal{C}'_1)^\perp,\bm{y})$ code obtained by removing the non-trivial $Z$-stabilizer $\bm{\gamma_0}$.
	\end{definition}
	Since $U_Z$ preserves the original CSS code, we have
	\begin{align}
		\sum_{\bm{\gamma}\in \mathcal{C}_2^\perp /\mathcal{C}_1^\perp}|A_{\bm{0},\bm{\gamma}}|^2=1.
	\end{align}
	The condition 
	\begin{align}
		\sum_{\bm{\gamma}'\in \langle \mathcal{C}_2^\perp /\mathcal{C}_1^\perp,\bm{\gamma_0} \rangle} |A'_{\bm{0},\bm{\gamma}'}|^2
		&=\sum_{\bm{\gamma}\in  \mathcal{C}_2^\perp /\mathcal{C}_1^\perp} |A'_{\bm{0},\bm{\gamma}}|^2 + |A'_{\bm{0},\bm{\gamma}\oplus\bm{\gamma_0}}|^2 \nonumber \\
		&=1.
	\end{align}
	is both necessary and sufficient for admissibility. Note that the induced logical operator \eqref{eqn:dig_log_op} corresponding to the trivial syndrome remains a diagonal unitary after splitting. 
	
	It is natural to ask how many $Z$-stabilizers are needed to determine a stabilizer code fixed by a given family of diagonal physical operators $U_Z$. Liang, Hu et al. \cite{coherent_noise} derived necessary and sufficient conditions for all transversal $Z$-rotations to preserve the codespace of a stabilizer code. The conditions require the weight $2$ $Z$-stabilizers to cover all the qubits that are in the support of the $X$-component of some stabilizer. Rengaswamy et al. \cite{rengaswamy2020optimality} derived less restrictive necessary and sufficient conditions for a single transversal $T$ gate. 
	
	The difference $A'_{\bm{0},\bm{\gamma}}-A'_{\bm{0},\bm{\gamma}\oplus\bm{\gamma_0}}$ depends on the new $X$-logical $\bm{w_0}$. For $\bm{\gamma}\in\mathcal{C}_2^\perp/\mathcal{C}_1^\perp$, let
	\begin{align}\label{eqn:s_g(w)}
		s_{\bm{\gamma}}(\bm{w_0})\coloneqq \frac{1}{|\mathcal{C}_1|}\sum_{\bm{u}\in \mathcal{C}_1+\bm{w_0}} (-1)^{\bm{\gamma}\bm{u}^T} d_{\bm{u}\oplus\bm{y}}.
	\end{align}
	It then follows from \eqref{eqn:GC_C1_y} that
	\begin{align}\label{eqn:split_w_s_1}
		A'_{\bm{0},\bm{\gamma}} 
		&= \frac{1}{2|\mathcal{C}_1|}\sum_{\bm{u}\in \langle\mathcal{C}_1,\bm{w_0}\rangle } (-1)^{\bm{\gamma}\bm{u}^T} d_{\bm{u}\oplus\bm{y}} \nonumber\\
		&= \frac{1}{2}\left( A_{\bm{0},\bm{\gamma}}+ s_{\bm{\gamma}}(\bm{w_0})\right),
	\end{align}
	and follows from \eqref{eqn:group_split_GCs} that 
	\begin{align}\label{eqn:split_w_s_2}
		A'_{\bm{0},\bm{\gamma}\oplus\bm{\gamma_0}} = \frac{1}{2}\left(A_{\bm{0},\bm{\gamma}} - s_{\bm{\gamma}}(\bm{w_0}))\right).
	\end{align}
	The quantity $s_{\bm{\gamma}}(\bm{w_0})$ determines whether or not a split is admissible.

	We design extensible splittings by expanding diagonal operators in the Pauli basis, and we illustrate our approach by constructing $Z^{1/2^l} \otimes Z^{1/2^l}$ from $Z^{1/2^{l-1}}$. 
	We write 
	\begin{align}
		Z^{1/2^{l-1}} \equiv c_l I+s_l Z,
	\end{align}
	where $c_l \coloneqq \cos \pi/2^l$ and $s_l \coloneqq -\imath\sin \pi/2^l$. Figure \ref{fig:Z-rot_split_schemes}(a) shows how we construct 
	\begin{align}
		\left(Z^{1/2^l}\right)^{\otimes 2} \equiv c^2_{l+1} I\otimes I &+ c_{l+1}s_{l+1}(I \otimes Z+Z \otimes I) \nonumber \\
		& + s^2_{l+1} Z\otimes Z,
	\end{align}
	\begin{figure*}[!ht]
		\centering
		\begin{tikzpicture}
			\tikzmath{\y2 =0; } 
			\node () at (2,11-\y2) {\footnotesize $I$};
			\node () at (6.5,11-\y2) {\footnotesize $Z$};
			\node (0) at (2,10.5-\y2) {\small $c_1(0) =\frac{1+x_{l-1}}{2}$};
			\node (1) at (6.5,10.5-\y2) {\small $c_1(1) = \frac{1-x_{l-1}}{2}$};
			
			\node () at (0,8.75-\y2) {\footnotesize $II$};
			\node () at (3,8.75-\y2) {\footnotesize $IZ$};
			\node () at (6,8.75-\y2) {\footnotesize $ZI$};
			\node () at (8.75,8.75-\y2) {\footnotesize $ZZ$};
			\node (00) at (0,9.25-\y2) {\small $c_2(00)=\frac{c_1(0)+1}{2}$};
			\node (01) at (3,9.25-\y2) {\small $c_2(01)=\frac{c_1(1)}{2}$};
			\node (10) at (6,9.25-\y2) {\small $c_2(10)=\frac{c_1(0)-1}{2}$};
			\node (11) at (8.75,9.25-\y2) {\small $c_2(11)=\frac{c_1(1)}{2}$};
			
			\node () at (-1,7-\y2) {\footnotesize $III$};
			\node () at (0.5,7-\y2) {\footnotesize $IIZ$};
			\node () at (2,7-\y2) {\footnotesize $IZI$};
			\node () at (3.5,7-\y2) {\footnotesize $IZZ$};
			\node () at (5,7-\y2) {\footnotesize $ZII$};
			\node () at (6.5,7-\y2) {\footnotesize $ZIZ$};
			\node () at (8,7-\y2) {\footnotesize $ZZI$};
			\node () at (9.5,7-\y2) {\footnotesize $ZZZ$};
			\node (000) at (-1,7.5-\y2) {\small $\frac{c_2(00)+1}{2}$};
			\node (001) at (0.5,7.5-\y2) {\small $\frac{c_2(01)}{2}$};
			\node (010) at (2,7.5-\y2) {\small $\frac{c_2(10)}{2}$};
			\node (011) at (3.5,7.5-\y2) {\small $\frac{c_2(11)}{2}$};
			\node (100) at (5,7.5-\y2) {\small $\frac{c_2(01)}{2}$};
			\node (101) at (6.5,7.5-\y2) {\small $\frac{c_2(11)}{2}$};
			\node (110) at (8,7.5-\y2) {\small $\frac{c_2(00)-1}{2}$};
			\node (111) at (9.5,7.5-\y2) {\small $\frac{c_2(10)}{2}$};
			
			\draw[->,red] (0) -- (00); 
			\draw[->,red] (0) -- (10); 
			\draw[->,red] (1) -- (01); 
			\draw[->,red] (1) -- (11);
			\draw[->,red] (00) -- (000);
			\draw[->,red] (10) -- (001);
			\draw[->,red] (01) -- (010);
			\draw[->,red] (11) -- (011);
			\draw[->,red] (00) -- (110);
			\draw[->,red] (10) -- (111);
			\draw[->,red] (01) -- (100);
			\draw[->,red] (11) -- (101);
			\node (Z) at (12,11-\y2) {\footnotesize $Z^{1/2^{l-1}}$};
			\node (CZ) at (12,9.75-\y2) {\footnotesize C$Z^{1/2^{l-1}}$};
			\node (CCZ) at (12,8.5-\y2) {\footnotesize CC$Z^{1/2^{l-1}}$};
			\node (CnZ) at (12,7.25-\y2) {\footnotesize C$^{(j)}Z^{1/2^{l-1}}$};
			\draw[->] (Z) to (CZ);
			\draw[->] (CZ) to (CCZ);
			\draw[->] (12,7.6-\y2) to (12,7.5-\y2);
			\node (1dots) at (12,8-\y2) {\vdots};
			\node () at (13,10.5-\y2) {\small \color{blue}$j$ is odd};
			\node () at (13,9.25-\y2) {\small \color{blue}$j$ is even};
			\node () at (13.5,8.2-\y2) {\small \color{blue}Hadamard};
			\node () at (13.5,7.7-\y2) {\small \color{blue}Construction};
		\end{tikzpicture}
		\caption{Admissible Splits from C$^{(j-1)}Z^{1/2^{l-1}}$ to C$^{(j)}Z^{1/2^{l-1}}$ for any fixed $l\ge 1$.}
		\label{fig:control_split}
	\end{figure*}
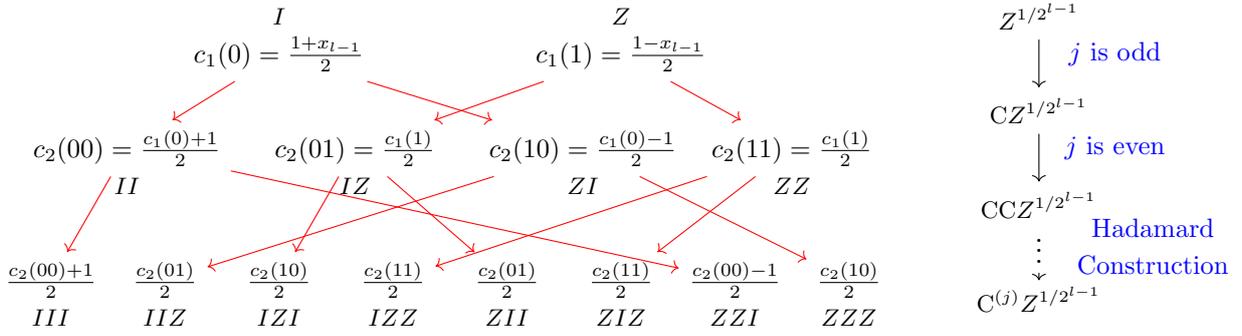
	by making use of the double angle formulas
	\begin{align}
		c_l = c^2_{l+1} +s^2_{l+1}   \text{ and } s_{l} = 2c_{l+1}s_{l+1}.
	\end{align}
	Recall that generator coefficients coincide with Pauli coefficients of the induced logical operator as described in \eqref{eqn:dig_log_op}. The splitting rule determines the values $s_{\bm{\gamma}}(\bm{w_0})$ needed to satisfy in \eqref{eqn:split_w_s_1} and \eqref{eqn:split_w_s_2}. Here we require
	\begin{align}\label{eqn:constraint_on_s}
		s_{\bm{\gamma}}(\bm{w_0}) = 
		\left\{ \begin{array}{lc}
			1, ~~ \text{ if } {\bm{\gamma}} = \bm{0}, \\
			0, ~~ \text{ if } {\bm{\gamma}} \neq \bm{0},
		\end{array}\right.
	\end{align}
	since we can write double-angle formulas as
	\begin{align}
		c^2_{l+1} = \frac{1}{2}\left(c_l +1\right), ~s^2_{l+1} &= \frac{1}{2}\left(c_l -1\right), 
		\\ \text{ and } s_{l+1}c_{l+1} &= \frac{1}{2}\left(s_l+0\right).
	\end{align}

	Note that this design only connects a single level in the Clifford hierarchy to the next level, that it does not extend indefinitely. In Figure \ref{fig:Z-rot_split_schemes}(b), we generalize the design to make it extend indefinitely. We include the global phase $x_{l}\coloneqq e^{\imath\pi/2^l}$ this time, and  decompose part of $x_l$ using the Euler's formula
	\begin{align}\label{eqn:euler}
		x_{l} =x_{l+1}x_{l+1}=x_{l+1}(c_{l+1}-s_{l+1}).  
	\end{align}
	Note that $Z^{1/2^{l-1}} = x_{l}(c_l I+s_l Z)$ and $\left(Z^{1/2^{l}}\right)^\dagger =\frac{x_{l+1}}{x_l}(c_{l+1} I-s_{l+1} Z)$. Then after splitting, we obtain the gate in one level higher
	\begin{align}
		Z^{1/2^{l-1}} \otimes \left(Z^{1/2^{l}}\right)^\dagger =& x_{l+1}(c_lc_{l+1}I\otimes I-c_l s_{l+1} I\otimes Z \nonumber\\
		&+ s_l c_{l+1}Z\otimes I -s_ls_{l+1}Z\otimes Z).
	\end{align}
	The decomposition in \eqref{eqn:euler} holds for any $l$, and we can use induction to prove that after splitting $j$ times, we obtain the gate
	\begin{align}
		Z^{1/2^{l-1}} \otimes \left(Z^{1/2^{l}}\right)^\dagger \otimes \cdots \otimes \left(Z^{1/2^{l-1+j}}\right)^\dagger.
	\end{align}
	Because of the non-uniform rotations, the values $s_{\bm{\gamma}}(\bm{w_0})$ needed to satisfy vary from step to step. We now introduce a splitting that is indefinitely extensible with simple requirement for $s_{\bm{\gamma}}(\bm{w_0})$.

	The diagonal operator C$^{(j-1)}Z^{1/2^{l-1}}=$ diag$[\bm{d_j}]$ for 
	\begin{align}
		\bm{d_j} = [\bm{1}_{2^{j-1}},\bm{1}_{2^{j-1}-1},x_{l-1}]^T,
	\end{align}
	where $\bm{1}_{m}$ is the all-one vector with length $m$. Let $\bm{e_1},\dots,\bm{e_{2^{j}}}$ be the standard basis of $\F_2^{2^{j}}$.We expand C$^{(j-1)}Z^{1/2^{l-1}}$ in the Pauli basis using the Walsh-Hadamard matrix $H_{2^j}$,
	\begin{align}
		\mathrm{C}^{(j-1)}Z^{1/2^{l-1}} = \sum_{\bm{v}\in\F_2^j} c_j(\bm{v})E(\bm{0,\bm{v}}),
	\end{align}
	where $\bm{c_{j}}\coloneqq[c_{j}(\bm{v})]_{\bm{v}\in\F_2^{j}}$ is given by 
	\begin{align}\label{eqn:recur_coeff_control}
		\bm{c_{j}} 
		= H_{2^{j}}\bm{d_{j}}  &=H_{2^{j}}\left(\bm{1}_{2^j}+\left(x_l-1\right)\bm{e_{2^j}}\right)\nonumber \\
		&= \bm{e_1} + \left(\frac{x_l-1}{2^{j}}\right)[(-1)^{w_H(\bm{v})}]^T_{\bm{v}\in\F_2^{j}}.
	\end{align}
	The recursive construction for the Walsh-Hadamard matrix leads to a recursion for the coefficients $c_j(\bm{v})$,
	\begin{align}
		\bm{c_{j+1}} &= \frac{1}{2}
		\begin{bmatrix}
			H_{2^{j-1}} & H_{2^{j-1}}\\
			H_{2^{j-1}} & -H_{2^{j-1}}
		\end{bmatrix}
		\begin{bmatrix}
			\bm{1}_{2^j}\\
			\bm{d_{j}} 
		\end{bmatrix}
	\end{align}
	so that
	\begin{align}
		\label{eqn:split_control_1}
		c_{j+1}([0,\bm{v}]) = \left(e_1\right)_{\bm{v}} + \left(\frac{x_l-1}{2^{j+1}}\right)(-1)^{w_H(\bm{v})} ,
	\end{align}
	and 
	\begin{align}
		\label{eqn:split_control_2}
		c_{j+1}([1,\bm{v}]) = - \left(\frac{x_l-1}{2^{j+1}}\right)(-1)^{w_H(\bm{v})}.
	\end{align}
	Here $\bm{e_1}=[\left(e_1\right)_{\bm{v}}]_{\bm{v}\in\F_2^{2^j}}$. Note that $w_H(\bm{v})+w_H(\bm{1}_j\oplus\bm{v})=j$. If $j$ is odd, then $(-1)^{w_H(\bm{v})}=-(-1)^{w_H(\bm{1}_j\oplus\bm{v})}$ and
	\begin{align}
		c_j(\bm{v}) =c_{j+1}([0,\bm{v}]) + c_{j+1}([1,\bm{1}_j\oplus\bm{v}]).
	\end{align}
	Let $\bm{t}= [0,\dots,0,1] \in \F_2^j$. If $j$ is even, then $(-1)^{w_H(\bm{v})}=-(-1)^{w_H(\bm{1}_j\oplus\bm{v}\oplus\bm{t})}$ and
	\begin{align}
		c_j(\bm{v}) =c_{j+1}([0,\bm{v}]) + c_{j+1}([1,\bm{1}_j\oplus\bm{v}\oplus\bm{t}]).
	\end{align}
	Figure \ref{fig:control_split} describes the splitting process of the cases $j=1,2$. 
	
	It then follows from \eqref{eqn:recur_coeff_control}, \eqref{eqn:split_control_1} and \eqref{eqn:split_control_2} that the requirement for $s_{\bm{\gamma}}(\bm{w_0})$ is the same as in \eqref{eqn:constraint_on_s}.  Although they share the same splitting rule, the global phase $x_{l}$ they differ becomes a local phase after splitting since $s_{\bm{\gamma}=\bm{0}}=1\neq 0$.
	
	\begin{figure}[h!]
		\centering
		\begin{tikzpicture}
			\tikzmath{\y1 =0; }
			\node () at (0.7,3.5-\y1) {$l$th: $Z^{\frac{1}{2^{l-1}}}$, C$Z^{\frac{1}{2^{l-2}}}$, $\dots$, C$^{(l-1)}Z$};
			\node () at (0,4.5-\y1) {3rd: $T = Z^{\frac{1}{4}}$, C$P$, CC$Z$};
			\node () at (-0.5,5.5-\y1) {2nd: $P=\sqrt{Z}$, C$Z$};
			\node () at (-1.4,6.5-\y1) {1st: $Z$};
			\draw[->] (-1,6.3) to (-1,5.7);
			\draw[->] (-1,5.3) to (-1,4.7);
			\draw[->,dashed] (-1,4.25) to (-1,3.7);
			\draw[->] (-0.9,6.3) to (0.75,5.75);
			\draw[->] (-0.9,5.3) to (0.6,4.7);
			\draw[->] (0.85,5.3) to (1.55,4.7);
			\draw[->,dashed] (0.7,4.25) to (1.6,3.6);
			\draw[->,dashed] (1.7,4.25) to (2.5,3.75);
		\end{tikzpicture}
		\caption{Admissible Splits among the Elementary Operators in the Diagonal Clifford Hierarchy.}
		\label{fig:elem_DCH}
	\end{figure}
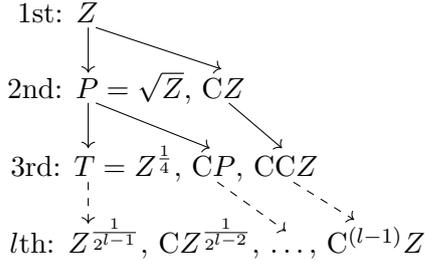
	Note that the admissible splits we describe include all the elementary operators in the diagonal Clifford hierarchy as shown in Figure \ref{fig:elem_DCH}. Figure \ref{fig:Z-rot_split_schemes} corresponds to the vertical line in Figure \ref{fig:elem_DCH}, and Figure \ref{fig:control_split} corresponds to the oblique line in Figure \ref{fig:elem_DCH}. 
	
	We now describe how to choose the new $X$-logical $\bm{w_0}$ to lift the level of the induced logical operator. For $l\ge 1$ we suppose that the physical transversal $Z$-rotation $\left(\exp{(-\imath\frac{\pi}{2^l})Z}\right)^{\otimes n}$ preserves an $\llbr n,k,d \rrbr $ CSS$(X,\mathcal{C}_2;Z,\mathcal{C}_1^\perp,\bm{y}=\bm{0})$ code, inducing a single $Z^{1/2^{l-1}}$ or C$^{(j)}Z^{1/2^{l-1}}$.
	
	\begin{theorem}
		\label{thm:specif_rem_Z}
		Suppose that after concatenation, the removal of $Z$-stabilizers introduces the new $X$-logical $\bm{w_0}=[\bm{1}_n,\bm{0}_n]$.
		
		\noindent Then, the logical operator lifts to $\left(Z^{1/2^{l}}\right)^{\otimes 2}$ or $\mathrm{C}^{(j)}Z^{1/2^{l-1}}$.
	\end{theorem}
	\begin{proof}
		Concatenation transforms the physical operator
		\begin{align}
			U_Z = \left(\exp{\left(-\imath\frac{\pi}{2^l}Z\right)}\right)^{\otimes n}\equiv \left(Z^{1/2^{l-1}}\right)^{\otimes n}
		\end{align} 
		into
		\begin{align}
			U'_Z = \left(\exp{\left(-\imath\frac{\pi}{2^{l+1}}Z\right)}\right)^{\otimes 2n} 
			\equiv \left(Z^{1/2^{l}}\right)^{\otimes 2n}.
		\end{align} 
		The physical operator $U'_Z$ preserves the $\llbr 2n,k,d'\ge d \rrbr$ CSS$(X,\mathcal{C}'_2;Z,(\mathcal{C}_1')^\perp,\bm{y}'=[\bm{0}_n,\bm{0}_n])$ codespace, as shown in Figure \ref{fig:concatenation} and Theorem \ref{thm:concatenation}. After concatenation, every element in $\mathcal{C}_1'$ takes the form $[\bm{u},\bm{u}]$ for some $\bm{u}\in \mathcal{C}_1$. Since $\bm{w_0}=[\bm{1}_n,\bm{0}_n]\notin \mathcal{C}_1'$, we can introduce $\bm{w_0}$ as a new $X$-logical ($\mathcal{C}''_1 = \langle \mathcal{C}'_1,\bm{w_0}\rangle$). 
		Concatenation does not change the generator coefficients, and it follows from \cite[Lemma 4]{generator_coeff_framework} that
		\begin{align}
			d_{[\bm{u},\bm{u}]}=\left(e^{-\imath\frac{\pi}{2^{l+1}}}\right)^{2n-2w_H([\bm{u},\bm{u}])}
		\end{align}
		for $[\bm{u},\bm{u}]\in\mathcal{C}'_1$. Let $\bm{\gamma}\in \mathcal{C}_2^\perp/\mathcal{C}_1^\perp$. Then $[\bm{\gamma},\bm{0}] \in  (\mathcal{C}'_2)^\perp/(\mathcal{C}'_1)^\perp$, and it follows form \eqref{eqn:s_g(w)} that
		\begin{align}
			&s_{[\bm{\gamma},\bm{0}]}([\bm{1},\bm{0}]) \nonumber \\
			&=\frac{1}{|\mathcal{C}'_1|}\sum_{[\bm{1}\oplus\bm{u},\bm{u}]\in\mathcal{C}'_1+[\bm{1},\bm{0}]}(-1)^{[\bm{\gamma},\bm{0}] [\bm{1}\oplus\bm{u},\bm{u}]^T} d_{[\bm{1}\oplus\bm{u},\bm{u}]\oplus[\bm{0},\bm{0]}} \nonumber \\
			&=\frac{1}{|\mathcal{C}_1|}\sum_{\bm{u}\in\mathcal{C}_1}(-1)^{\bm{\gamma}(\bm{1}\oplus\bm{u})^T}\left(e^{-\imath\frac{\pi}{2^{l+1}}}\right)^{2n-2w_H([\bm{1}\oplus\bm{u},\bm{u}])}. 
		\end{align}
		Since $w_H([\bm{1}\oplus\bm{u},\bm{u}])=n$ for all $\bm{u}\in\mathcal{C}_1$, we have 
		\begin{align}
			s_{[\bm{\gamma},\bm{0}]}([\bm{1},\bm{0}])
			&=(-1)^{\bm{\gamma}\bm{1}^T}\frac{1}{|\mathcal{C}_1|}\sum_{\bm{u}\in\mathcal{C}_1}(-1)^{\bm{\gamma}\bm{u}^T}\nonumber\\
			&= \left\{\begin{array}{lc}
				1, & \text{ if } \bm{\gamma}=\bm{0}, \\
				0, & \text{ if } \bm{\gamma}\neq\bm{0},
			\end{array} \right.
		\end{align}
		and the theorem now follows from \eqref{eqn:constraint_on_s}.
	\end{proof}

	\addtocounter{example}{-2}
	\begin{example}[Continued: from $\llbr 14,1,3\rrbr$ to $\llbr 14,2,2\rrbr$; Logical $P^\dagger \to (T^\dagger)^{\otimes 2}$ ]
		\label{examp:log_from_steane}
		\normalfont
		The $\llbr 14,1,3\rrbr$ code is obtained by concatenating the $\llbr 7,1,3\rrbr$ Steane code. 
		We introduce the new $X$-logical $\bm{w_0}=[\bm{1},\bm{0}]\in\F_2^{2n}$ by removing the $Z$-stabilizer $\bm{\gamma_0}=[\bm{1},\bm{1}]\in(\mathcal{C}'_1)^\perp$ to produce the $\llbr 14,2,2\rrbr$ code. The generator coefficients $A''_{\bm{0},\bm{\gamma}^{''}}\left(\frac{\pi}{4}\right)$ of the $\llbr 14,2,2\rrbr$ code for $\bm{\gamma}''\in\langle \bm{\gamma_1} = [\bm{1},\bm{0}],\bm{\gamma_0} \rangle$ under the physical $T^{\otimes 14}$ gate are
		\begin{align}
			A''_{\bm{0},\bm{\gamma}^{''}=\bm{0}}\left(\frac{\pi}{4}\right) &=\frac{1}{2}\left(\cos\frac{\pi}{4}+1\right) = \left(\cos\frac{\pi}{8}\right)^2,\nonumber \\
			A''_{\bm{0},\bm{\gamma}^{''}=\bm{\gamma_1}}\left(\frac{\pi}{4}\right) &= \frac{1}{2}\imath\sin\frac{\pi}{4} 
			= \imath\sin\frac{\pi}{8}\cos\frac{\pi}{8}.
		\end{align}
		Splitting gives
		\begin{align}
			A''_{\bm{0},\bm{\gamma}^{''}=\bm{\gamma_0}} 
			&= A'_{\bm{0},\bm{0}} - A^{''}_{\bm{0},\bm{\gamma}^{''}=\bm{0}} = \left(\imath\sin\frac{\pi}{8}\right)^2,\nonumber\\
			A''_{\bm{0},\bm{\gamma}^{''}=\bm{\gamma_1}\oplus\bm{\gamma_0}} 
			&= A'_{\bm{0},\bm{\gamma_1}} - A''_{\bm{0},\bm{\gamma}^{''}=\bm{\gamma_1}} = \imath\sin\frac{\pi}{8}\cos\frac{\pi}{8}.
		\end{align}
		It follows from \eqref{eqn:dig_log_op} that the logical operator induced by $T^{\otimes 14}$ on the $\llbr 14,2,2\rrbr$ codespace is $\left(T^\dagger\right)^{\otimes 2}$. Note that the $\llbr 14,2,2\rrbr$ code is a member of the triorthogonal code family introduced by Bravyi and Haah \cite{bravyi2012magic}. The operations described above can transform the $\llbr 15,1,3 \rrbr$ triorthogonal code \cite{knill1996accuracy,bravyi2005universal,landahl2013complex,anderson2014fault} to the $\llbr 30,2,2 \rrbr$ code for which the physical transversal $\sqrt{T}$ induces a logical $\sqrt{T}^\dagger$. The same operations work for the whole punctured Reed-Muller family $\llbr 2^{l+1}-1,1,3\rrbr $ \cite{landahl2013complex} that realize the single logical $Z^{1/2^{l-1}}\in\mathcal{C}_d^{(l)}$ and results in the $\llbr 2^{l+2}-2,2,2\rrbr$ triorthogonal code family realizing the logical transversal $Z^{1/2^{l}}\in\mathcal{C}_d^{(l+1)}$. 
	\end{example}
	
	\begin{example}[Continued: the $\llbr 2^l,l,2\rrbr$ code family realizes C$^{(l-1)}Z$]
		\label{examp:log_from_422}
		\normalfont
		Starting from the $\llbr 4,2,2 \rrbr$ code, we first concatenate to obtain the $\llbr 8,2,2\rrbr$ code, and then remove the $Z$-stabilizer associated with adding the new $X$-logical $\bm{w_0}=[\bm{1},\bm{0}]$ to produce the $\llbr 8,3,2\rrbr$ code. The $\llbr 4,2,2 \rrbr$ code realizes C$^{(1)}Z=$C$Z$ up to some logical Pauli $Z$ by either physical transversal Phase gate $P^{\otimes 4}$ or transversal Control-$Z$ gate C$Z^{\otimes 2}$. The $\llbr 8,3,2 \rrbr$ code realizes C$^{(2)}Z=$CC$Z$ up to some logical Pauli $Z$ by either physical transversal T gate $T^{\otimes 8}$ or transversal Control-Phase gate C$P^{\otimes 4}$. Repeated concatenation and removal of $Z$-stabilizers yields the $\llbr 2^l,l,2\rrbr$ code family that supports the logical C$^{(l-1)}Z$ gate up to some logical Pauli $Z$. When the physical gate is a transversal $Z$-rotation, the generator coefficients of the $\llbr 2^l,l,2\rrbr$ code family are listed below. 
		\begin{table}[h!]
			\centering
			\caption{The Splitting of Generator Coefficients for the induced logical C$^{(l-1)}$Z (up to some logical Pauli $Z$). The $\llbr 2^l,l,2\rrbr$ CSS codes are preserved by physical transversal $Z$-rotations $ \left(\exp{\left(-\imath\frac{\pi}{2^{l-1}}Z\right)}\right)^{\otimes 2^l}$.}
			\renewcommand{\arraystretch}{1.3} 
			\begin{tabular}{|c|c|c|}
				\hline
				&
				$U_Z^L$ up to $Z^L$
				&
				Generator Coefficients $A_{\bm{0},\bm{\gamma}}$
				\\
				\hline
				2 & 
				C$^{(1)}Z$ &
				$\frac{1}{2}$ $-\frac{1}{2}$ $-\frac{1}{2}$ $-\frac{1}{2}$ \\
				\hline
				3 & 
				C$^{(2)}Z$ &
				$\frac{3}{4}$ $-\frac{1}{4}$ $-\frac{1}{4}$ $\cdots$ $-\frac{1}{4}$ $-\frac{1}{4}$ \\
				\hline
				$l$ & 
				C$^{(l-1)}$Z &
				$\frac{2^{l-1}-1}{2^{l-1}}$ $-\frac{1}{2^{l-1}}$ $-\frac{1}{2^{l-1}}$ $\cdots$ $-\frac{1}{2^{l-1}}$ \\
				\hline
			\end{tabular}
			\label{tab:CC..CZ_splitting}
		\end{table}
	\end{example}
	
	Since removing $Z$-stabilizers may decrease code distance, we introduce a third elementary operation in the next Section with the aim of increasing the distance. 
	
	\begin{figure*}
		\centering
		\begin{tikzpicture}
			\draw[dashed,black] (6.5,0) -- (6.5,6.3);
			\hspace{-290pt}
			\node (b) at (10.75,6) {(a) Adding an $X$-stabilizer};
			\node (Z) at (12,0) {$\{ \bm{0} \}$};
			\node (C2) at (12,1.75) {$\mathcal{C}_2$};
			\node (C1) at (12,3.5) {$\mathcal{C}_1$};
			\node (F2m) at (12,5.25) {$\mathbb{F}_2^{n}$};
			\draw[<-,black,dotted,thick] (11.8,0.875) to [out=180,in=180] (11.8,2.625);g
			\node (x0) at (11,1.75) {$\color{blue}\bm{x_0}$};
			\path[draw] (Z) -- (C2)  -- (C1) -- (F2m);
			
			\node (Zp) at (14.5,0) {$\{ \bm{0} \}$};
			\node (C1p) at (14.5,1.75) {$\mathcal{C}_1^{\perp}$};
			\node (C2p) at (14.5, 3.5) {$\mathcal{C}_2^{\perp}$};
			\node (F2m) at (14.5,5.25) {$\mathbb{F}_2^{n}$};
			\draw[<-,black,dotted,thick] (14.7,4.375) to [out=-10,in=10] (14.7,2.625);
			\node (mu0) at (15.5,3.5) {$\color{blue}\bm{\mu_0}$};
			\path[draw] (Zp) -- (C1p) -- (C2p)  -- (F2m); 
			
			\hspace{510pt}
			\draw[] (0,4) -- (3,4) -- (3,2) -- (0,2) -- (0,4);
			\draw[] (0,3.7) -- (3,3.7);
			\draw[blue,thick] (1.5,4) rectangle (3,2); 
			\draw[fill=red] (1.5,4) rectangle (1.8,3.7); 
			\draw[fill=green] (0,4) rectangle (1.5,3.7); 
			\node (gamma) at (1.5,4.3) {$\bm{\gamma}\in \mathcal{C}_2^\perp / \mathcal{C}_1^\perp$};
			\node (mu) at (-0.2,3) {$\bm{\mu}$};
			\node (A_0,gamma) at (2.3,3.5) {\color{red}\small$A_{\bm{0},\bm{\gamma'}\oplus\bm{\mu_0}}$};
			\draw[] (5.5,5) -- (7,5) -- (7,1) -- (5.5,1) -- (5.5,5);
			\draw[] (5.5,3) -- (7,3); 
			\draw[] (5.5,4.7) -- (7,4.7);
			\draw[] (5.5,2.7) -- (7,2.7);
			\draw[fill=red] (5.5,2.7) rectangle (5.8,3); 
			\draw[fill=green] (5.5,5) rectangle (7,4.7); 
			\draw[blue,thick] (5.5,1) rectangle (7,3); 
			\node (A'_0,gamma) at (7,2.5) {\color{red}\small$A'_{\bm{\mu_0},\bm{\gamma}'} = A_{\bm{0},\bm{\gamma'}\oplus\bm{\mu_0}}$}; 
			\node (gamma') at (6.2,5.3) {$\bm{\gamma}'\in\langle \mathcal{C}_2,{\color{blue}\bm{x_0}}\rangle^\perp / \mathcal{C}_1^\perp$};
			\node (mu) at (5.3,4) {$\bm{\mu}$};
			\node (mu') at (4.8,2) {$\bm{\mu}+{\color{blue}\bm{\mu_0}}$};
			
			\node (eqn) at (3.2,6) {(b) Transforming the table of generator coefficients}; 
			\draw[->] (3.5,3.15) to (4.5,3.15);
			\node (rm) at (4,2.6) {remove};
			\draw[<-] (3.5,2.85) to (4.5,2.85);
			\node (ad) at (4,3.4) {add};
		\end{tikzpicture}
		\caption{(a) Adding the old $X$-logical $\bm{x_0}$ as a new $X$-stabilizer transforms the old $Z$-logical $\bm{\mu_0}$ to a new $X$-syndrome. (b) Introducing a new $X$-stabilizer $\bm{x_0}$ doubles the number of $X$-syndromes and halves the number of $Z$-logicals. The blue rectangle shifts as the generator coefficients evolve. }
		\label{fig:add_X_stab}
	\end{figure*}
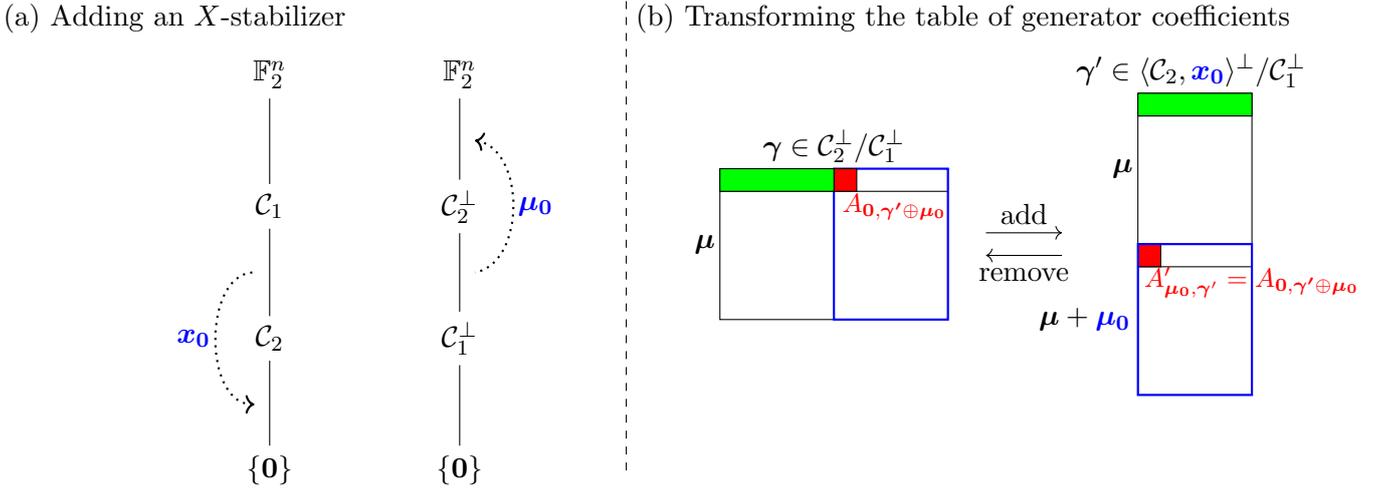
	
	\section{Increase Distance}
	\label{sec:add_X_stab}
	Our focus on diagonal gates $U_Z$ that preserve CSS$(X,\mathcal{C}_2;Z,\mathcal{C}_1^\perp,\bm{y})$ codes implies that the effective distance is the $Z$-distance, $d_Z = \min_{\bm{z}\in \mathcal{C}_2^\perp \setminus \mathcal{C}_1^\perp}w_H(\bm{z})$. Concatenation, described in Figure \ref{fig:concatenation}, does not change $d_Z$. Removal of $Z$-stabilizers increases the number of $Z$-logicals in $\mathcal{C}_2^\perp \setminus \mathcal{C}_1^\perp$, and this may decrease $d_Z$. After removing $Z$-stabilizers we may need to increase effective distance by introducing new $X$-stabilizers. We now examine how generator coefficients evolve when we add or remove $X$-stabilizers.    
	
	Adding a new $X$-stabilizer $\bm{x_0}\in \mathcal{C}_1\setminus\mathcal{C}_2$ transforms a CSS($X$, $\mathcal{C}_2$; $Z$, $\mathcal{C}_1^\perp$, $\bm{y}$) code to a CSS($X$, $\langle \mathcal{C}_2, \bm{x_0} \rangle$; $Z$, $\mathcal{C}_1^\perp,\bm{y}$) code. A $Z$-logical $\bm{\mu_0}$ in the original code becomes an $X$-syndrome in the new code. Note that $\bm{\mu_0}\in \mathcal{C}_2^\perp \setminus\mathcal{C}_1^\perp$ and $\bm{\mu_0}\notin \langle \mathcal{C}_2, \bm{x_0} \rangle^\perp \setminus \mathcal{C}_1^\perp$. The number of $Z$-logicals is halved, while the number of $X$-syndromes is doubled, so the number of generator coefficients remains constant. Let $U_Z$ be a fixed diagonal physical gate. The generator coefficients $A_{\bm{\mu},\bm{\gamma}}$ for the old code determine the generator coefficients $A'_{\bm{\mu}',\bm{\gamma}'}$ for the new code as follows:
	\begin{align*}
		A'_{\bm{\mu}',\bm{\gamma}'} 
		& = \sum_{\bm{z}\in \mathcal{C}_1^\perp + \bm{\mu}'+ \bm{\gamma}'} \epsilon_{(\bm{0},\bm{z})} f(\bm{z}) \\
		& = 
		\left\{\begin{array}{lc}
			A_{\bm{\mu}',\bm{\gamma}'}, 
			& \text{ if } \bm{\mu}' \in \F^n_2 / \mathcal{C}_2^\perp, \\
			A_{\bm{\mu}'\oplus\bm{\mu_0},\bm{\gamma}'\oplus \bm{\mu_0}}, 
			& \text{ if } \bm{\mu}' \oplus \bm{\mu_0} \in \F^n_2 / \mathcal{C}_2^\perp.
		\end{array} \right.\numberthis\label{eqn:reshape_GCs}
	\end{align*}
	Note that the new $Z$-logical $\bm{\gamma}'\in \langle \mathcal{C}_2, \bm{x_0} \rangle^\perp/\mathcal{C}_1^\perp$. If $\bm{\mu}'$ coincides with an old syndrome, then $A'_{\bm{\mu}',\bm{\gamma}'} = A_{\bm{\mu}',\bm{\gamma}'}$. Otherwise $\bm{\mu}'\oplus\bm{\mu_0} \in \F^n_2 / \mathcal{C}_2^\perp$ and $\bm{\gamma}'\oplus\bm{\mu_0} \in \mathcal{C}_2^\perp /\mathcal{C}_1^\perp$. Figure \ref{fig:add_X_stab} captures the process of adding and removing $X$-stabilizers. Note that \eqref{eqn:reshape_GCs} is reversed when an $X$-stabilizer is removed. 
	
	If we remove an $X$-stabilizer from a CSS code that is preserved by a diagonal gate $U_Z$, then the new code is still preserved by $U_Z$. If instead, we add an $X$-stabilizer,   then the new code may fail to be preserved by $U_Z$. We say that addition of an $X$-stabilizer is \emph{admissible} if the new code is preserved by $U_Z$. We now characterize admissible additions in terms of the new $X$-syndrome $\bm{\mu_0}$. 
	
	Let $\mathcal{C}_2^\perp/\mathcal{C}_1^\perp = \langle D,\bm{\mu_0} \rangle$. The old is preserved by $U_Z$ if and only if 
	\begin{align}
		\sum_{\bm{\gamma}\in \langle D,\bm{\mu_0}\rangle}|A_{\bm{0},\bm{\gamma}}|^2=1,
	\end{align}
	and the new code is preserved by $U_Z$ if and only if
	\begin{align}
		\sum_{\bm{\gamma}\in D}|A_{\bm{0},\bm{\gamma}}|^2=1.
	\end{align}
	Addition of $\bm{x_0}$ is admissible if and only if
	\begin{align}
		A_{\bm{0},\bm{\gamma}} = 0 \text{ for all } \bm{\gamma}\in D+\bm{\mu_0}.
	\end{align}
	We require that half the generator coefficients $A_{\bm{0},\bm{\gamma}}$ vanish. The non-vanishing coefficients appear in the green rectangle shown in Figure \ref{fig:add_X_stab}(b). Then, it follows from \eqref{eqn:dig_log_op} that the logical operator stays at the same level after an admissible addition. It also follows from \eqref{eqn:split_w_s_1} and \eqref{eqn:split_w_s_2} that an addition is admissible if and only if 
	\begin{align}
		s_{\bm{\gamma}}(\bm{w_0})=\pm A_{\bm{0},\bm{\gamma}} \text{ for all } \bm{\gamma}\in \mathcal{C}_2^\perp/\mathcal{C}_1^\perp.
	\end{align}
	We may need to concatenate several times and remove several independent $Z$-stabilizers to create enough zeros among the generator coefficients. 
	
	We now combine concatenation, removal of $Z$-stabilizers, and addition of $X$-stabilizers to construct a CSS code family with growing distance that is preserved by diagonal operators with increasing logical level in the Clifford hierarchy.
	
	\begin{example}[Quantum Reed-Muller (QRM) Code Family]
		\label{examp:QRM_fam}
		\normalfont
		Introduced in \cite[Theorem 14]{generator_coeff_framework} and \cite[Theorem 19]{rengaswamy2020optimality}, this is a family of $\llbr 2^m, \binom{m}{r},2^{\min\{r,m-r\}}\rrbr$ CSS codes preserved by physical transversal $Z$-rotations $\left(Z^{1/2^{\left({m}/{r}-1\right)}}\right)^{\otimes 2^m}$
		when $r \mid m$. We now describe how these codes are constructed by concatenation followed by removal of $Z$-stabilizers and addition of $X$-stabilizers. 
		
		Let $r\ge 1$ be fixed. Note that $m/r$ increases by $1$ when $m$ increases by $r$, and that the new code is preserved by a physical gate that is one level higher in the Clifford hierarchy. We start from a $\llbr 2^m,\binom{m}{r},2^{\min\{r,m-r\}}\rrbr$ CSS code determined by $\mathcal{C}_1=\mathrm{RM}(r,m)$ and $\mathcal{C}_2=\mathrm{RM}(r-1,m)$. The recursive construction of classical Reed-Muller codes \cite{macwilliams1977theory} is given by 
		\begin{align}\label{eqn:recur_RM}
			\mathrm{RM}(r,m+1) = \{(\bm{u},\bm{u}\oplus\bm{v})\mid &\bm{u}\in \mathrm{RM}(r,m), \nonumber \\
			&\bm{v}\in \mathrm{RM}(r-1,m)\}.
		\end{align}
		Let $\bm{1}_{2^r}$ denotes the vector of length $2^r$ with every entry equals to 1. We concatenate our CSS code $r$ times to construct the $\llbr 2^{m+r}, \binom{m}{r},2^{\min\{r,m-r\}} \rrbr$ CSS code determined by $\mathcal{C}'_1=\bm{1}_{2^r}\otimes \mathrm{RM}(r,m)$ and $\mathcal{C}'_2=\bm{1}_{2^r}\otimes\mathrm{RM}(r-1,m)$. 
		Note that $\mathcal{C}'_1 \subseteq \mathrm{RM}(r,m+r)$ and $\mathcal{C}'_2 \subseteq \mathrm{RM}(r-1,m+r)$. We now remove the $Z$-stabilizers and add the $X$-stabilizers to make  $\mathcal{C}'_1=\mathrm{RM}(r,m+r)$, $\mathcal{C}'_2=\mathrm{RM}(r-1,m+r)$. We obtain the $\llbr 2^{m+r},\binom{m+r}{r},2^{\min\{r,m\}}\rrbr $ CSS code which is the next member of the QRM code family. The level of the new induced logical operator equals that of the new physical transversal $Z$-rotations \cite[Theorme 19]{rengaswamy2020optimality}, which is one level higher than that of the old induced logical operator. For fixed $r$, the operations described above just maintain the distance. 
		
		To achieve the growing distance, we can increase $r$ by $1$, and increase $m$ by $h\coloneqq r+\frac{m}{r}+1$ so that $\frac{m}{r}+1 = \frac{m+h}{r+1}$. When $r\mid m$, it follows from \eqref{eqn:recur_RM} that we can obtain the $\llbr 2^{m+h},\binom{m+h}{r+1},2^{\min{\{r+1,m+h-r-1\}}} \rrbr$ CSS code from a $\llbr 2^m,\binom{m}{r},2^{\min\{r,m-r\}}\rrbr$ CSS code by first concatenating $h$ times, then removing $\left( \binom{m+h}{r+1}+\binom{m+h}{r}-\binom{m}{r}\right)$ $Z$-stabilizers, and adding $\binom{m+h}{r}$ $X$-stabilizers. The logical operator induced by the new code is one level higher than that of the old code, and the distance doubles for the new code. Figure $\ref{fig:three_elem_ops}$ illustrates the case when $m=2$ and $r=1$.
	\end{example}
	
	\section{Conclusion} 
	\label{sec:conc}
	Given a CSS code that realizes a diagonal gate at the $l^{\mathrm{th}}$ level, we have introduced three basic operations that can be combined to construct a new CSS code that realizes a diagonal gate at the $(l+1)^{\mathrm{th}}$ level in the Clifford hierarchy. The three basic operations are concatenation (to increase the physical level), removal of $Z$-stabilizers (to increase the logical level and increase code rate), and addition of $X$-stabilizers (to increase the distance). We have derived necessary and sufficient conditions for admissibility, that is for the new code to be preserved by the target physical operator. We have described these conditions using the mathematical framework of generator coefficients. Concatenation is always admissible, while the other two basic operations may not be admissible. We have demonstrated the power of combining the three basic operations to synthesize a target diagonal operator by climbing the Clifford hierarchy to construct the QRM code family.
	
	In future work, we expect to explore how best to balance removal of $Z$-stabilizers and addition of $X$-stabilizers. We will also investigate the existence of code families corresponding to Figure \ref{fig:Z-rot_split_schemes}(b). 
	\section*{Acknowledgement}
	The work of the authors was supported in part by NSF under grant CCF1908730.
	\bibliographystyle{IEEEtran}

\end{document}